\documentclass{article}

\usepackage[dvipsnames]{xcolor}
\usepackage{algorithm}
\usepackage{algpseudocode}

\usepackage{color}
\usepackage{fancyvrb}
\usepackage{listings,natbib}
\usepackage{verbatim}
\usepackage{graphicx}
\usepackage{amsmath}
\usepackage{geometry}
\usepackage{bbm}
\usepackage{amsthm}
\usepackage{capt-of}
\usepackage{enumitem}
\setlist{leftmargin=15pt,labelindent=15pt}
\setlist[enumerate]{align=left}

\renewcommand{\th}{{\theta}}
\newcommand{\Th}{{\Theta}}

\newcommand{\calN}{{\mathcal N}}
\newcommand{\calP}{{\mathcal P}}
\newcommand{\calF}{{\mathcal F}}
\newcommand{\calX}{{\mathcal X}}

\newcommand{\hsp}{{\hspace{2ex}}}
\newcommand{\by}{{\mathbf y}}

\newcommand{\bw}{{\mathbf w}}

\newcommand{\one}{{\mathbf 1}}

\newcommand{\calZ}{{\mathcal Z}}

\newcommand{\zn}{{\mathcal Z_N}}
\newcommand{\z}{{\zeta}}

\newcommand{\eps}{{\epsilon}}
\newcommand{\al}{{\alpha}}

\def\RR{{\mathbf{ R}}}

\newcommand{\beq}{\begin{equation}}
\newcommand{\eeq}{\end{equation}}
\newcommand{\beqn}{\begin{eqnarray}}
\newcommand{\eeqn}{\end{eqnarray}}

\newcommand{\pr}{\mbox{Pr}}

\newtheorem{theorem}{Theorem}[section]
\newtheorem{definition}{Definition}[section]
\newtheorem{corollary}{Corollary}[theorem]
\newtheorem{lemma}[theorem]{Lemma}

\title{Finding our Way in the Dark: Approximate MCMC for Approximate Bayesian Methods}
\author{
Evgeny Levi \\
\and
Radu V. Craiu \thanks{Email: craiu@utstat.toronto.edu} 
}
\date{\small Department of Statistical Sciences, University of Toronto}

\begin{document}

\maketitle

\begin{abstract}
With larger amounts of data at their disposal, scientists are emboldened to tackle  complex questions that require sophisticated statistical models. It is not unusual for the latter to have  likelihood functions that elude analytical formulations. Even under such adversity, when one can simulate from the sampling distribution,  Bayesian analysis can be conducted using approximate methods such as Approximate Bayesian Computation (ABC) or Bayesian Synthetic Likelihood (BSL).  A significant drawback of these  methods is that the number of required simulations can be prohibitively large, thus severely  limiting their scope.  In this paper we design perturbed MCMC samplers that can be used within the ABC and BSL paradigms to significantly accelerate computation while maintaining control on computational efficiency. The  proposed strategy relies on recycling  samples from the chain's past.  The  algorithmic design is supported by a theoretical analysis while  practical performance is examined via a series of simulation examples and data analyses.  
\end{abstract}

\section{Introduction}   
\label{sec:intro}



 Since the early 1990s Bayesian statisticians have been able to operate largely due to the rapid development of Markov chain Monte Carlo (MCMC) sampling methods \citep[see, for example][for a recent review]{craiu2014bayesian}.   Given observed
data $\by_0 \in \calX^n$ with sampling density $f(\by_0|\th)$  indexed by parameter $\theta\in \RR^q$, Bayesian inference for functions of $\th$ rely on the characteristics of the  posterior distribution
\beq
\pi(\th|\by_0)=\frac{p(\th)f(\by_0|\th)}{\int_{R^q} p(\th)f(\by_0|\th) d\th}\propto p(\th)f(\by_0|\th),
\label{eq:post}
\eeq
where  $p(\th)$ denotes  the prior distribution. 
 When the posterior distribution in \eqref{eq:post}  cannot be studied analytically,  we rely on MCMC algorithms  to generate  samples from  $\pi$.  While traditional MCMC samplers such as Metropolis-Hastings or Hamiltonian MCMC \citep[see][and references therein]{brooks2011handbook} can sample distributions with unknown normalizing constants, they rely on the closed form of the unnormalized posterior,  $p(\th)f(\by_0|\th)$. 

The advent of large data has altered  in multiple ways the framework we just described. For example, larger data tend to yield likelihood functions that are much more expensive to compute, thus exposing the liability inherent in the iterative nature of MCMC samplers. In response to this challenge, new computational methods  based on divide and conquer  \citep{cons,weierstrass,entezari2018likelihood}, subsampling \citep{bardenet,quiroz2}, or sequential  \cite{balakrishnan2006one,maclaurin2015firefly} strategies have emerged. Second, it is understood that larger data should yield answers to more complex problems.   This implies the  use of increasingly complex models, in as much as the sampling distribution is no longer available in closed form. 

In the absence of a tractable likelihood function, statisticians have developed  approximate methods to perform Bayesian inference when, for any parameter value $\th \in \RR^q$,  data  $\by \sim f(\by|\th)$ can be sampled from the model. Here we consider two alternative approaches that have been proposed and gained considerable momentum in recent years:  the Approximate Bayesian Computation (ABC) \citep{marin2012approximate, baragatti2014overview, sisson2018overview, drovandi2018abc} and the Bayesian Synthetic Likelihood (BSL)\citep{wood2010statistical, drovandi2018approximating, price2018bayesian}.  
Both algorithms are effective when they are combined with  Markov chain Monte Carlo sampling schemes to produce samples from an approximation of $\pi$ and both share the  need for generating many pseudo-data sets $\by \sim f(\by|\th)$. This comes with serious challenges when the data is large and generating a pseudo-data set  is computationally expensive. In this paper we tackle the reduction of computational burden by recycling  draws from the chain's  history. While this reduces drastically the computation time, it alters the transition kernel of the original  MCMC chain. We demonstrate that we can control the approximating error introduced when perturbing the original  kernel using some of the error analysis for perturbed Markov chains developed recently by
\cite{mitrophanov2005sensitivity}, \cite{johndrow2015optimal} and \cite{johndrow2017error}.

 The paper is structured as follows. Section~\ref{sec:review} briefly reviews the ABC method and  Section~\ref{sec:meth}   introduces the proposed MCMC algorithms for ABC. Section~\ref{sec:bsl} reviews  BSL sampling and extends the proposed methods to this class of approximations. The practical impact of these algorithms is evaluated  via  simulations in Section~\ref{sec:sim} and data analyses in Section~\ref{sec:real}. The theoretical analysis showing control of perturbation errors in total variation norm is in  Section~\ref{sec:theory}. The paper closes with conclusions and ideas for  future work. 
\section{Approximate Bayesian Computation}
 \label{sec:review}
In order to illustrate the ABC sampler, let us consider the following Hidden Markov Model (HMM) 
\begin{eqnarray}
X_0 \sim & p(x_0), \nonumber \\
X_i|x_{i-1} \sim & p(X_i|x_{i-1},\th), \hsp i=1,\ldots, n, \label{trans}\\
Y_i|x_i \sim & p(Y_i|x_i,\th),\hsp i=1,\ldots, n .\label{emiss}
\end{eqnarray} 
Unless Gaussian distributions are used to specify  the transition and emission laws given in \eqref{trans} and \eqref{emiss}, respectively, the marginal distribution 
$P(y_1,\cdots,y_n|\th)$ cannot be calculated in closed form. It is possible to treat the hidden random variables $X_i$ as auxiliary and sample them using Particle MCMC (PMCMC) \citep{andrieu2010particle} or ensemble MCMC \citep{shestopaloff2013mcmc}. However, computations become increasingly difficult as $n$ increases. Moreover, for some financial time series models  such as Stochastic Volatility for log return, the $\al$-Stable distribution may be useful to model transition and/or emission probabilities \citep{nolan2003modeling}. The challenge is that the stable distributions do not have closed form densities, thus rendering the particle and ensemble MCMC impossible to use. Other widely used  examples where the likelihood functions cannot be expressed analytically include various networks models \citep[e.g.,][]{kolaczyk2014statistical} and Markov random fields \citep{rue2005gaussian}. 
For such models with intractable or computationally expensive likelihood evaluations, simulation based algorithms such as ABC are frequently used for  inference. In its simplest form, the  ABC is an
 accept/reject sampler.  Given a user-defined summary statistic $S(\by) \in \RR^p$, the Accept/Reject ABC is described in Algorithm~\ref{alg:abc-sim}.
\begin{algorithm}
\caption{Accept/Reject ABC}
\label{alg:abc-sim}
\begin{algorithmic}[1]
\State Given observed $\by_0$ and required number of samples $M$.
\For{$t=1,\cdots,M$}
\State Match = FALSE
\While{Not Match}
\State $\z^* \sim p(\z)$
\State $\by \sim f(\by|\z^*)$
\If{$S(\by) = S(\by_0)$}
\State $\th^{(t)}=\z^*$,
\State Match = TRUE.
\EndIf
\EndWhile
\EndFor
\end{algorithmic}
\end{algorithm}

 We emphasize that a  closed form equation for the likelihood is not needed, only the ability to generate from $f(\by|\th)$ for any $\th$. If $S(\by)$ is a sufficient statistics and $\pr(S(\by)=S(\by_0))>0$ then the algorithm yields posterior samples from the true posterior $\pi(\th|\by_0)$. Alas, the level of complexity for   models where ABC is needed, makes it unlikely for these two conditions to hold.    In order to implement ABC under more realistic assumptions, a  (small) constant $\eps$ is chosen and  $\z^*$ is accepted whenever $d(S(\by),S(\by_0))<\eps$, where  $d(S(\by),S(\by_0))$ is a user-defined distance function.  
The introduction of $\eps>0$ and the use of non-sufficient statistics remove layers of exactness from the target distribution. The approximating distribution is  denoted  $\pi_\eps(\th|S(\by_0))$ and we have 
\beq
\lim_{\eps\downarrow 0} \pi_\eps(\th|S(\by_0))=\pi(\th|S(\by_0)).
\label{abc:conv}
\eeq
In light of \eqref{abc:conv} one would like to have $S(\by)=\by$, but if the sample size of  $\by_0$ is large, then the curse of dimensionality leads to $\pr(d(\by,\by_0)<\epsilon) \approx 0$. Thus, obtaining even a moderate number of samples using ABC can be an unattainable goal in this case.  In almost all cases of interest, $S$ is not a sufficient statistics, implying that some information about $\th$ is lost. Not surprisingly,  much attention has been focused on finding appropriate low-dimensional summary statistics for inference \citep[see, for example][]{robert2011lack,  fearnhead2012constructing, marin2014relevant, prangle2015summary}. In this paper we assume that the summary statistic  $S(\by)$ is given. 

In the absence of information about the model parameters, the prior and posterior distributions may be misaligned, having non-overlapping regions of mass concentration. Hence,  parameter values that are drawn from the prior will be rarely retained making the algorithm very inefficient. 
\begin{algorithm}
\caption{ABC MCMC}
\label{alg:abc_mcmc}
\begin{algorithmic}[1]
\State Given $\by_0$, $\eps>0$ and required number of samples $M$.
\State Find initial $\th^{(0)}$ with simulated $\by$ such that $d(S(\by),S(\by_0))<\eps$.
\For{$t=1,\cdots,M$}
\State Generate $\z^* \sim q(\cdot|\th^{(t-1)}).$
\State Simulate $\by^* \sim f(\by|\z^*)$ and let $\delta(\by^*,\by_0) = d(S(\by^*),S(\by_0)).$
\State Calculate $\alpha = \min\left(1,\one_{\{\delta(\by^*,\by_0)<\eps\}}\frac{p(\z^*)q(\th^{(t-1)}|\z^*)}{p(\th^{(t-1)})q(\z^*|\th^{(t-1)})}\right)$
\State Generate independent $U \sim \mathcal U(0,1).$
\If{$U\leq \alpha$}
\State $\th^{(t)}=\z^*$.
\Else
\State $\th^{(t)}=\th^{(t-1)}$.
\EndIf 
\EndFor
\end{algorithmic}
\end{algorithm}  
Algorithm~\ref{alg:abc_mcmc} presents the ABC-MCMC algorithm of 
\cite{marjoram2003markov} which avoids sampling from the prior and instead relies on building a chain with a Metropolis-Hastings  (MH) transition kernel, with state space $\{(\th,\by) \in  \RR^q \times \calX^n\}$, proposal distribution $ q(\z|\th)\times f(\by|\z)$ and  target distribution 
\begin{equation}
\pi_{\eps}(\th,\by|\by_0)\propto p(\th)f(\by|\th)\one_{\{\delta(\by_0,\by)<\eps\}},
\label{abc:t2}
\end{equation} 
where $\delta(\by_0,\by) = d(S(\by),S(\by_0))$.
Note that the goal is the marginal distribution for $\th$ which is:
\begin{equation}
\pi_\eps(\th|\by_0)=\int \pi_{\eps}(\th,\by|\by_0)d\th \propto \int p(\th)f(\by|\th)\one_{\{\delta(\by_0,\by)<\eps\}}d\th = p(\th)\pr(\delta(\by_0,\by)<\eps|\th).
\label{abc:trg}
\end{equation} 

There are  a few alternatives  to    Algorithm~\ref{alg:abc_mcmc}.    For instance, \cite{lee2012discussion} approximates $P(\delta(\by_0,\by)<\eps|\th)$ via one of its unbiased estimators, $J^{-1} \sum_{j=1}^J 1_{\{\delta(\by_0,\by_j)<\eps\}}$ where $J\ge 1$ and each $\by_j$ is simulated from $f(\by| \th)$. The use of unbiased estimators for $P(\delta(\by_0,\by)<\eps|\th)$ when computing the MH acceptance  ratio can be validated using the theory  of   pseudo-marginal MCMC samplers \citep{andrieu2009pseudo}. Clearly, when the probability $P(\delta<\eps|\th)$ is very small, this method would require simulating a large number of $\delta$s (or equivalently $\by$s)  in order to move to a new state. 
Other MCMC designs suitable for ABC  can be found in \cite{bornn2014one}. 

Sequential Monte Carlo (SMC) samplers have also been successfully used for ABC (henceforth denoted  ABC-SMC)  \citep{sisson2007sequential, lee2012choice, filippi2013optimality}.  
 ABC-SMC requires a specified decreasing sequence $\eps_0>\cdots>\eps_J$. Lee's method \cite{lee2012choice} uses the Particle MCMC design \citep{andrieu2010particle} in which samples are updated as the target  distribution evolves with $\eps$. More specifically, it starts by sampling $\th_0^{(1)},\ldots, \th_0^{(M)}$ from $\pi_{\eps_0}(\th|\by_0)$ using Accept-Reject ABC. Subsequently, at time $t+1$ all samples are sequentially updated  so their  distribution is $\pi_{\eps_{t+1}}(\th|\by_0)$ \citep[see][for a complete description]{lee2012choice}. The advantage of this method is not only that it starts from large $\eps$, but also that it generates independent draws. 
 A comprehensive coverage of computational techniques for ABC can be found in \cite{sisson2018handbook} and references therein.
 We also note a general lack of guidelines concerning the selection of $\eps$, which is unfortunate as the performance of ABC sampling  depends heavily on its value. 
To make a fair comparison between different methods, we revise ABC-MCMC algorithm by introducing a decreasing sequence $\eps_0>\cdots>\eps_J$ ($J$ is number of "steps") similar to ABC-SMC and "learning" transition kernel during burn-in as in Algorithm~\ref{alg:ABC-MCMC-M}.
\begin{algorithm}
\caption{ABC MCMC modified (ABC-MCMC-M)}
\label{alg:ABC-MCMC-M}
\begin{algorithmic}[1]
\State Given $\by_0$, sequence $\eps_0>\cdots>\eps_J$, constant $c$, burn-in period $B$ and required number of samples $M$.
\State Define $\eps = \eps_0$.
\State Find initial $\th^{(0)}$ with simulated $\by$ such that $d(S(\by),S(\by_0))<\eps$.
\State Let $\tilde \mu$ be expectation of prior distribution and $\tilde \Sigma=c\Sigma$ where $\Sigma$ is covariance of the prior $p(\th)$.
\State Define, $b=\lfloor(B/J)\rfloor$ and define sequence $(a_1,\cdots,a_J)=(b,2b,\cdots,Jb)$.
\For{$t=1,\cdots,M$}
\If{$t=a_j$ for some $j=1,\cdots, J$}
\State Set $\eps = \eps_j$.
\State Find $\tilde \mu$ as mean of $\{\th^{(t)}\}$ $t=1,\cdots,(a_j-1)$ and $\tilde \Sigma=c\Sigma$ where $\Sigma$ is covariance of $\{\th^{(t)}\}$ $t=1,\cdots,(a_j-1)$.
\EndIf
\State Generate $\z^* \sim q(\cdot|\th^{(t-1)},\tilde \mu,\tilde \Sigma).$
\State Simulate $\by^* \sim f(\by|\z^*)$ and let $\delta^* = d(S(\by^*),S(\by_0)).$
\State Calculate $\alpha = \min\left(1,\one_{\{\delta^*<\eps\}}\frac{p(\z^*)q(\th^{(t-1)}|\z^*,\tilde \mu,\tilde \Sigma)}{p(\th^{(t-1)})q(\z^*|\th^{(t-1)},\tilde \mu,\tilde \Sigma)}\right)$
\State Generate independent $U \sim \mathcal U(0,1).$
\If{$U\leq \alpha$}
\State $\th^{(t)}=\z^*$.
\Else
\State $\th^{(t)}=\th^{(t-1)}$.
\EndIf 
\EndFor
\end{algorithmic}
\end{algorithm}  
Since the choice of proposal distribution $q(\cdot|\th)$  can considerably influence the performance of ABC-MCMC, we consider finite adaptation during the burn-in period of length $B$. In addition, during burn-in  the $\eps$ also varies, starting with a higher value (which makes it easier to find the initial $\th^{(0)}$ value) and gradually decreasing in accordance to a pre-determined scheme.  In our implementations we use independent MH sampling or RWM. In the former case, the  proposal is Gaussian $\calN(\cdot|\tilde{\mu},\tilde{\Sigma})$ with $c= 3$. The RWM proposal  is $\calN(\cdot|\th^{(t-1)},\tilde{\Sigma})$ with $c=2.38^2/q$  \citep{roberts1997weak, roberts2001optimal}. 

All the algorithms discussed so far rely on numerous generations of pseudo-data. Since the latter can be computationally costly,  proposals for reducing the simulation cost are made in \cite{wilkinson2014accelerating} and \cite{jarvenpaa2018efficient}.  The approaches are based on learning the dependence between  $\delta$ and  $\th$ and, from it,  establishing directly whether a proposal $\th$ should be accepted or not. Flexible regression models are used to model these unknown functional relationships. The overall  performance depends  on the signal to noise ratio and on the model's performance in  capturing  patterns that can be highly complex. 

To accelerate ABC-MCMC we consider a different approach and propose to store and utilize past simulations (with appropriate weights) in order to speed up the calculation while keeping under control the resulting approximating errors. The objective is to approximate $P(\delta<\eps|\z^*)$ for any $\z^*$ at every MCMC iteration using past simulated $(\z,\delta)$ proposals, making the whole procedure computationally faster. The changes proposed perturb the chain's transition kernel 
and we rely on the theory developed  by \cite{mitrophanov2005sensitivity} and \cite{johndrow2015approximations} to assess the approximating error for the posterior. 
The  k-Nearest-Neighbor (kNN) method is used to integrate past observations into the transition kernel. The main advantage of kNN is that it is uniformly strongly consistent which guarantees that for a large enough chain history, we can control the error between the intended stationary distribution and that of the proposed accelerated MCMC as  shown in Section~\ref{sec:theory}.  
\section{Approximated ABC-MCMC (AABC-MCMC)}
\label{sec:meth}
 In this section we describe an  ABC-MCMC algorithm that utilizes past  simulations to significantly improve computational efficiency.  As noted previously, the ABC-MCMC with threshold $\eps$ targets the density 
\begin{equation}
\label{eq:abc_true}
\pi_{\eps}(\th|\by_0)\propto p(\th)P(\delta(\by_0,\by)<\eps|\th),
\end{equation}
where $\delta(\by_0,\by)=d(S(\by),S(\by_0))$ with $\by\sim f(\by|\th)$ and $\th \in \Th$. Denote $h(\th):= P(\delta(\by_0,\by)<\eps|\th)$ and note that if $h$ were known for every $\th$ then we could run an  MH-MCMC chain with invariant target density proportional to $p(\th)h(\th)$.    
Alas,  $h$ is almost always unknown and  unbiased estimates can be computationally expensive or statistically inefficient.  We build an alternative approach that relies on  consistent estimates of $h$ that rely on the chain's past history, are much cheaper to compute, and require a new theoretical treatment.  

To fix ideas, suppose that at time $t$ we generate the proposal  $(\zeta_{t+1}, \bw_{t+1}) \sim q(\zeta|\theta^{(t)}) f(\bw|\zeta)$ and suppose  that at iteration $N$,  all the proposals $\zeta_n$, regardless whether they were accepted or rejected, along with corresponding distances $\delta_n=\delta(\bw_n,\by_0)$ are available for $0\le n \le N-1$. This past history is stored in the set $\calZ_{N-1}=\{\zeta_n,\delta_n\}_{n=1}^{N-1}$. Given a new proposal $\zeta^* \sim q(|\th^{(t)})$, we  generate $\bw^* \sim f( \cdot |\zeta^*)$  and compute $\delta^*=d(S(\bw^*),S(\by_0))$. Set $\zeta_N=\zeta^*$, $\bw_N=\bw^*$,  $\calZ_{N}=\calZ_{N-1} \cup \{(\zeta_N,\delta_N)\}$ and  estimate $h(\zeta^*)$ using 
\begin{equation}
 \hat h(\zeta^*) = \frac{ \sum_{n=1}^{N}W_{Nn}(\zeta^*) \one_{\{\delta_n<\eps\}} }{\sum_{n=1}^{N}W_{Nn}(\zeta^*)},
\label{hw2}
\end{equation} 
where  $W_{Nn}(\zeta^*)=W(\|\zeta_n-\zeta^*\|)$ are weights and  $W:\RR \rightarrow[0,\infty)$ is a decreasing function. We  discuss a couple of choices for the function $W(\cdot)$  below. 

{\it Remark 1:} Note that if some of the past proposals have been accepted, then  the Markovian property of the chain is violated since the acceptance probability  does not depend solely on the current state, but also on the past ones. We defer the  theoretical considerations for dealing with adaptation in the context of perturbed Markov chains to a future communication. Below,  we modify slightly the construction above while respecting the core idea. 

In order to separate the samples used as proposals from those used to estimate $h$ in \eqref{hw2}, we will generate at each time $t$ two independent samples $\zeta_{t+1}\sim q(\zeta|\theta^{(t)})$ and $(\tilde \zeta_{t+1}, \tilde \bw_{t+1})$ from $q(\zeta|\theta^{(t)}) f(\bw|\zeta)$. Then, the history  $\calZ$ collects the $(\tilde \zeta,\tilde \delta)$ samples while the proposal used to update the chain is the $\zeta$ sample. With this notation \eqref{hw2} becomes
\begin{equation}
 \hat h(\zeta^*) = \frac{ \sum_{n=1}^{N}W_{Nn}(\zeta^*) \one_{\{\tilde \delta_n<\eps\}} }{\sum_{n=1}^{N}W_{Nn}(\zeta^*)},
\end{equation} 
where $\tilde \delta_n = \delta(\tilde \bw ,\by_0)$ and  $W_{Nn}(\zeta^*)=W(\|\tilde \zeta_n-\zeta^*\|)$.

\begin{algorithm}
\caption{Approximated ABC MCMC (AABC-MCMC)}
\label{alg:AABC}
\begin{algorithmic}[1]
\State Given $\by_0$ with summary statistics $s_0=S(\by_0)$, sequence $\eps_0>\cdots>\eps_J$, constant $c$, burn-in period $B$, required number of samples $M$, initial simulations $\calZ_N=\{\tilde \zeta_n,\tilde \delta_n\}_{n=1}^{N}$ with $\tilde \zeta_n \sim p(\zeta)$, $\tilde \bw_n \sim f(\cdot|\tilde \zeta_n)$ and $\tilde \delta_n = d(S(\tilde \bw_n),s_0)$.
\State Define $\eps = \eps_0$.
\State Find initial $\th^{(0)}$ with simulated $\by$ such that $d(S(\by),s_0)<\eps$.
\State Let $\tilde \mu$ be expectation of prior distribution and $\tilde \Sigma=c\Sigma$ where $\Sigma$ is covariance of the prior $p(\th)$.
\State Define, $b=\lfloor(B/J)\rfloor$ and define sequence $(a_1,\cdots,a_J)=(b,2b,\cdots,Jb)$
\For{$t=1,\cdots,J$}
\If{$t=a_j$ for some $j=1,\cdots, J$}
\State Set $\eps = \eps_j$.
\State Find $\tilde \mu$ as mean of $\th^{(t)}$ $t=1,\cdots,(a_j-1)$ and $\tilde \Sigma=c\Sigma$ where $\Sigma$ is covariance of $\th^{(t)}$ $t=1,\cdots,(a_j-1)$.
\EndIf
\State Generate $\zeta^* \sim \calN(\cdot;\tilde \mu,\tilde \Sigma)$ and $\tilde \zeta^* \sim \calN(\cdot;\tilde \mu,\tilde \Sigma)$. 
\State Simulate $\tilde \bw^* \sim f(\cdot|\tilde \zeta^*)$ and let $\tilde \delta^* = d(S(\tilde \bw^*),s_0).$
\State Add the dual simulated pair of parameter and discrepancy to the past set: $\calZ_N = \calZ_{N-1}\cup \{\tilde \zeta^*,\tilde \delta^*\}$ and set $N=N+1$.
\State $\hat h(\zeta^*) = \frac{ \sum_{n=1}^{N}W_{Nn}(\zeta^*) \one_{\{\tilde \delta_n<\eps\}} }{\sum_{n=1}^{N}W_{Nn}(\zeta^*)}$.
\State $\hat h(\th^{(t)}) = \frac{ \sum_{n=1}^{N}W_{Nn}(\th^{(t)}) \one_{\{\tilde \delta_n<\eps\}} }{\sum_{n=1}^{N}W_{Nn}(\th^{(t)})}$.
\State Calculate $\alpha = \min\left(1,\frac{p(\zeta^*)\hat h(\zeta^*) \calN(\th^{(t)};\tilde \mu,\tilde \Sigma)}{p(\th^{(t)})\hat h(\th^{(t)})\calN(\zeta^*;\tilde \mu,\tilde \Sigma)}\right)$
\State Generate independent $U \sim \mathcal U(0,1).$
\If{$U\leq \alpha$}
\State $\th^{(t+1)}=\zeta^*$.
\Else
\State $\th^{(t+1)}=\th^{(t)}$.
\EndIf 
\EndFor
\end{algorithmic} 
\end{algorithm}  

{\it Remark 2:}  Even if $\delta^*$ is greater than $\eps$ (which would trigger automatically rejection for ABC-MCMC), suppose there is a close neighbour of $\zeta^*$ whose corresponding $\delta$ is less than $\eps$. Then the estimated $h(\z^*)$ will not be zero and there is a chance of moving to a different state. Intuitively, this is expected to reduce the variance of the accepting probability estimate. 

{\it Remark 3:}  When comparing the unbiased estimator 
\beq \tilde h(\zeta^*)={1\over K} \sum_{j=1}^K  \one_{\{\tilde \delta_j<\eps\}},
\label{unb}
\eeq with the  consistent estimator \beq
\hat h(\zeta^*) = \frac{ \sum_{n=1}^{N}W_{Nn}(\zeta^*) \one_{\{\tilde \delta_n<\eps\}} }{\sum_{n=1}^{N}W_{Nn}(\zeta^*)},
\label{wgh}\eeq
 we hope to outperform both the   small and large $K$ cases in \eqref{unb}. For the small $K$, we expect to reduce the variability in our acceptance probabilities, while for larger $K$ we expect  to reduce the computational costs without sacrificing much in terms of precision.

Since the proposed weighted estimate is no longer an unbiased estimator of $h(\th)$, a new theoretical evaluation is needed to study the effect of  perturbing  the  transition kernel on the statistical analysis. Central to the algorithm's utility is the ability to control the total variation distance between the desired distribution of interest given in \eqref{eq:abc_true} and the modified chain's target. As will be shown in Section~\ref{sec:theory}, we rely  on three assumptions to ensure that the chain would approximately sample from \eqref{eq:abc_true}: 1) compactness of $\Th$; 2) uniform ergodicity of the chain using the true $h$ and 3) uniform convergence in probability of $\hat h(\th)$ to $h(\th)$ as $N\to \infty$. 

The k-Nearest-Neighbor (kNN) regression approach \citep{fix1951discriminatory, biau2015lectures} has a property of uniform consistency \citep{cheng1984strong}.  Define $K=g(N)$ (in our numerical experiments we have used $g(\cdot)=\sqrt{(\cdot)}$). Without loss of generality we relabel the elements of $\calZ_N=\{\tilde \zeta_n,\tilde \delta_n\}_{n=1}^N$ according to distance  $\| \tilde \zeta_n-\zeta^*\|$ so that $(\tilde \z_1,\tilde \delta_1)$   and $(\tilde \z_N,\tilde \delta_N)$ corresponds to the smallest and largest  among all distances $\{\|\tilde \zeta_j -\zeta^*\|: \; 1\le j \le N\}$, respectively. The kNN method sets $W_{Nn}(\z^*)$ to zero for all $n>K$.   For $n\leq K$, we focus on the following two weighting schemes:
\begin{enumerate}
\item[(U)]  The \textit{uniform} kNN  with $W_{Nn}(\z^*)=1$ for all $n \leq K$; 
\item[(L)] The \textit{linear} kNN  with $W_{Nn}(\zeta^*)=W(\|\tilde \zeta_n -\zeta^*\|)=1-\|\tilde \zeta_n -\zeta^*\|/\|\tilde \zeta_K-\zeta^*\|$ for $n\leq K$ so that the weight decreases from $1$ to $0$ as 
$n$ increases from $1$ to $K$.
\end{enumerate}

The kNN's  theoretical properties that are  used to validate our sampler rely on   independence between the  pairs  $\{\tilde \zeta_n,\tilde \delta_n\}_{n\ge1}$. Therefore, throughout the paper, we use an independent proposal in the MH sampler, i.e.  $q(\cdot|\th^{(t)})=q(\cdot)$  and $q$ is Gaussian. The entire procedure is outlined in Algorithm~\ref{alg:AABC}. 

 To conclude, at the end of a simulation of size $M$ the MCMC samples are $\{\th^{(1)},\ldots, \th^{(M)}\}$ and the history used for updating the chain is $\{(\tilde \zeta_1, \tilde\delta_1),\ldots, (\tilde \zeta_M, \tilde \delta_M)\}$. The two sequences are independent of one another, i.e.  for any $N>0$, the elements in $\calZ_N$ are independent of the chain's history up to time $N$. 
 
Note also that $h(\theta^{(t)})$ is required in order  to determine the  acceptance probability at step $t+1$. In this case the $h$-value may be updated  if $\|\th^{(t)} - \tilde \z^*\|$ is small enough. 


In the next section we extend the approximate MCMC construction  to Bayesian Synthetic Likelihood. In Sections 5 and 6 we use numerical experiments to show that the proposed procedure generally improves the mixing of a chain.  
\section{BSL and Approximated BSL (ABSL)}
\label{sec:bsl}

An alternative approach to bypass the intractability of the sampling distribution is proposed by  \cite{wood2010statistical}. His approach is based on the assumption that the conditional  distribution for a user-defined statistic $S(\by)$ given $\th$ is
Gaussian with mean $\mu_{\th}$ and covariance matrix $\Sigma_{\th}$. The  \textit{Synthetic Likelihood} (SL) procedure assigns to each  $\th$ the likelihood  $SL(\th)={\cal{N}}(s_{0};\mu_{\th},\Sigma_{\th})$, where  $s_0=S(\by_0)$ and
${\cal{N}}(x; \mu, \Sigma)$ denotes the density of a normal with mean $\mu$ and covariance $\Sigma$. SL can be used for maximum likelihood estimation as in \cite{wood2010statistical} or within the Bayesian paradigm as proposed by \cite{drovandi2018approximating} and \cite{price2018bayesian}.  The latter work proposes to sample the approximate posterior generated by the Bayesian Synthetic Likelihood (BSL) approach,  $\pi(\th|s_0)\propto p(\th)\calN(s_{0};\mu_{\th},\Sigma_{\th})$, using a MH sampler.   Direct calculation of the acceptance probability  is not possible because  the conditional mean and covariance are unknown for any $\theta$. However, both can be estimated based on $m$ statistics $(s_1,\cdots,s_m)$ sampled from their conditional distribution given $\th$. More precisely, after simulating $\by_i\sim f(\by|\th)$ and setting $s_i=S(\by_i)$, $i=1,\cdots,m$,  one can estimate
\begin{equation}
\label{eq:bsl}
\begin{split}
\hat \mu_{\th} & = \frac{\sum_{i=1}^m s_i}{m},\\
\hat \Sigma_{\th} & =  \frac{ \sum_{i=1}^m (s_i-\hat\mu_{\th})(s_i-\hat\mu_{\th})^T }{m-1},
\end{split}
\end{equation} 
so that the synthetic likelihood is 
\beq SL(\th|\by_0)={\cal{N}}(S(\by_0);\hat \mu_{\th} , \hat \Sigma_{\th}).
\label{sl}
\eeq
\begin{algorithm}
\caption{Bayesian Synthetic Likelihood (BSL-MCMC)}
\label{alg:bsl}
\begin{algorithmic}[1]
\State Given $s_0=S(\by_0)$, number of simulations $m$ and required number of samples $M$.
\State Get initial $\th^{(0)}$, estimate $\hat \mu_{\th^{(0)}},\hat \Sigma_{\th^{(0)}}$ by simulating $m$ statistics given $\th^{(0)}$.
\State Define $h(\th^{(0)}) = \calN(s_0;\hat \mu_{\th^{(0)}},\hat \Sigma_{\th^{(0)}})$.
\For{$t=1,\cdots,M$}
\State Generate $\zeta^* \sim q(\cdot|\th^{(t-1)}).$
\State Estimate $\hat \mu_{\zeta^*},\hat \Sigma_{\zeta^*}$ by simulating $m$ pseudo-data points $\{\by^{(j)}: \; 1\le j \le m\}$ and corresponding statistics $\{S(\by^{(j)}):\; 1\le j \le m\}$ given $\zeta^*$.
\State Calculate $h(\zeta^*) = \calN(s_0;\hat \mu_{\zeta^*},\hat \Sigma_{\zeta^*})$.
\State Calculate $\alpha = \min\left(1,\frac{p(\zeta^*)h(\zeta^*) q(\th^{(t-1)}|\zeta^*)}{p(\th^{(t-1)})h(\th^{(t-1)})q(\zeta^*|\th^{(t-1)})}\right)$
\State Independently generate $U \sim \mathcal U(0,1).$
\If{$U\leq \alpha$}
\State Set $\th^{(t)}=\zeta^*$ .
\Else
\State Set $\th^{(t)}=\th^{(t-1)}$.
\EndIf 
\EndFor
\end{algorithmic}
\end{algorithm}   
The pseudo-code in Algorithm \ref{alg:bsl} shows the steps involved in the BSL-MCMC sampler. Since each MH step requires calculating the likelihood ratios between two SLs calculated at different parameter values, one can anticipate the heavy computational load involved in running the chain for thousands of iterations, especially if sampling data  $\by$ is expensive.  Note that even though these estimates for the conditional mean and covariance are unbiased, the estimated value of the Gaussian likelihood is biased and therefore pseudo marginal MCMC theory is not applicable. \cite{price2018bayesian}  presented an unbiased Gaussian likelihood estimator and  have empirically showed that using biased and unbiased estimates generally perform similarly. They have also remarked that this procedure is very robust to the number of simulations $m$, and demonstrate empirically that using $m=50$ to $200$  produce similar results.  

The normality assumption for summary statistics is certainly a strong assumption which may not hold in practice. Following up on this, \cite{an2018robust} relaxed the jointly Gaussian assumption to Gaussian copula with non-parametric marginal distribution estimates (NONPAR-BSL), which includes joint Gaussian as a special case, but is much more flexible. The estimation is based, as in the BSL framework, on $m$ pseudo-data samples simulated for each $\th$. 

Clearly, BSL is computationally costly and requires many pseudo-data simulations to obtain Monte Carlo samples of even moderate sizes. To accelerate BSL-MCMC we propose to store and utilize past simulations of $(\z,s)$ to approximate $\mu_{\zeta^*},\Sigma_{\zeta^*}$ for any $\zeta^* \in \Th$, making the whole procedure computationally faster. As in the previous section, we separate the simulation used to update the chain from the simulation used to enrich the history of the chain.    
The approach can trivially be extended for NONPAR-BSL but we do not pursue it further here. K-Nearest-Neighbor (kNN) method is used as a non-parametric estimation tool for different quantities described above. As will be shown in Section~\ref{sec:theory} with the proposed method we can control the error between the intended stationary distribution and that of the proposed accelerated MCMC.
\subsection*{Approximated Bayesian Synthetic Likelihood (ABSL)}      
Setting $s_0=S(\by_0)$ and assuming conditional normally for this statistic the objective is to sample from
\begin{equation}
\label{eq:bsl_true}
\pi(\th|s_0)\propto p(\th)\calN(s_0;\mu_{\th},\Sigma_{\th}).
\end{equation}     
During the MCMC run, {the proposal $\zeta^*$ is generated from $q(\cdot)$ 
and the history $\calZ_N$ is enriched using 
$\tilde \zeta^* \sim q( \cdot)$, $\{ \tilde \by^{*(j)}\}_{j=1}^m \stackrel{iid}{\sim} f(\by |\tilde\zeta^*) $ and $\{\tilde s^{*(j)}=S(\tilde \by^{*(j)})\}_{j=1}^m$}. Then for any $\zeta$, the conditional mean and covariance of statistics vector is estimated using past samples as weighted averages:
\begin{equation}
\begin{split}
\hat \mu_{\zeta} & = \frac{\sum_{n=1}^N [ W_{Nn}(\zeta) \sum_{j=1}^m \tilde s_n^{(j)}]}{m\sum_{n=1}^N W_{Nn}(\zeta)},\\
\hat \Sigma_{\zeta} & =  \frac{ \sum_{i=1}^N [W_{Nn}(\zeta) \sum_{j=1}^m (\tilde s_n^{(j)}-\hat\mu_{\zeta})(\tilde s_n^{(j)}-\hat\mu_{\zeta})^T ]}{m\sum_{i=1}^N W_{Nn}(\zeta)}.
\end{split}
\end{equation}
\begin{algorithm}
\caption{Approximated Bayesian Synthetic Likelihood (ABSL)}
\label{alg:ABSL}
\begin{algorithmic}[1]
\State Given $s_0=S(\by_0)$, constant $c$, burn-in period $B$, $J$ number of adaption points during burn-in , required number of samples $M$, initial pseudo data simulations $\calZ_N=\{\tilde \zeta_n,\{\tilde s_n^{(j)}\}_{j=1}^m\}_{n=1}^{N}$ with $\tilde \zeta_n \sim p(\zeta)$, $\tilde \by_n^{(j)} \sim f(\by|\tilde \zeta_n)$ and $\tilde s_n^{(j)} = S(\tilde \by_n^{(j)})$.
\State Get initial $\th^{(0)}$.
\State Let $\tilde \mu$ be expectation of prior distribution and $\tilde \Sigma=c\Sigma$ where $\Sigma$ is covariance of the prior $p(\th)$.
\State Define, $b=\lfloor(B/J)\rfloor$ and define sequence $(a_1,\cdots,a_J)=(b,2b,\cdots,Jb)$
\For{$i=1,\cdots,M$}
\If{$i=a_j$ for some $j=1,\cdots, J$}
\State Find $\tilde \mu$ as mean of $\th^{(t)}$ $t=1,\cdots,(a_j-1)$ and $\tilde \Sigma=c\Sigma$ where $\Sigma$ is covariance of $\th^{(t)}$ $t=1,\cdots,(a_j-1)$.
\EndIf
\State Generate $\zeta^* , \tilde \zeta^* \stackrel{iid}{\sim} \calN(\cdot;\tilde \mu,\tilde \Sigma).$
\State Simulate $\tilde \by^{*(j)} \sim f(\by|\tilde \zeta^*)$ and let $\tilde s^{*(j)} = S(\tilde \by^{*(j)})$ for $1\le j \le m$. 
\State Add simulated parameter and statistics to the past set: $\calZ_N = \calZ_{N-1} \cup \{\tilde \zeta^*,\{\tilde s_n^{(j)}\}_{j=1}^m\}\}$ and set $N=N+1$.
\State Calculate: 
\begin{eqnarray*}
\hat \mu_{\zeta^*}  &=& \frac{\sum_{n=1}^N [ W_{Nn}(\zeta^*) \sum_{j=1}^m \tilde s_n^{(j)}]}{m\sum_{n=1}^N W_{Nn}(\zeta^*)}\\
\hat \Sigma_{\zeta^*}  &=&  \frac{ \sum_{i=1}^N [W_{Nn}(\zeta^*) \sum_{j=1}^m (\tilde s_n^{(j)}-\hat\mu_{\zeta^*})(\tilde s_n^{(j)}-\hat\mu_{\zeta^*})^T ]}{m\sum_{i=1}^N W_{Nn}(\zeta^*)}
\end{eqnarray*}
\State Calculate: 
\begin{eqnarray*}
\hat \mu_{\th^{(t)}} &=& \frac{\sum_{n=1}^N[ W_{Nn}(\th^{(t)}) \sum_{j=1}^m \tilde s_n^{(j)}]}{m\sum_{n=1}^N W_{Nn}(\th^{(t)})}\\ 
\hat \Sigma_{\th^{(t)}}  &=&  \frac{ \sum_{i=1}^N[ W_{Nn}(\th^{(t)}) \sum_{j=1}^m (\tilde s_n^{(j)}-\hat\mu_{\th^{(t)}})(\tilde s_n^{(j)}-\hat\mu_{\th^{(t)}})^T] }{m\sum_{i=1}^N W_{Nn}(\th^{(t)})}.
\end{eqnarray*}

\State $\hat h(\zeta^*) = \calN(s_0;\hat \mu_{\zeta^*},\hat \Sigma_{\zeta^*})$.
\State $\hat h(\th^{(t)}) = \calN(s_0;\hat \mu_{\th^{(t)}},\hat \Sigma_{\th^{(t)}})$.
\State Calculate $\alpha = \min\left(1,\frac{p(\zeta^*)\hat h(\zeta^*) \calN(\th^{(t)};\tilde \mu,\tilde \Sigma)}{p(\th^{(t)})\hat h(\th^{(t)})\calN(\zeta^*;\tilde \mu,\tilde \Sigma)}\right)$.
\State Generate independent $U \sim \mathcal U(0,1).$
\If{$U\leq \alpha$}
\State $\th^{(t+1)}=\zeta^*$.
\Else
\State $\th^{(t+1)}=\th^{(t)}$.
\EndIf 
\EndFor
\end{algorithmic} 
\end{algorithm}  Again the weights are functions of distance between proposed value and parameters' values from the past $W_{Nn}(\zeta)=W(\|\zeta -\tilde \zeta_n\|)$, where $\|\cdot \|$ is the Euclidean norm. To get appropriate convergence properties we use the kNN approach to calculate weights $W_{Nn}$, where only the $K=\sqrt{N}$ closest values to $\zeta$ are used in the calculation of conditional means and covariances. As  in the previous section, uniform (U) and linear (L) weights are used.  Once again we expect that the use of the chain's cumulated  history  can significantly speed up the whole procedure since it relieves the pressure to simulate many data sets $\by$ at every step.  The use of the independent Metropolis kernel ensures that $\calZ_N$  contains  independent draws  which is required for theoretical validation  in Section~\ref{sec:theory}. We will also show that under mild assumptions and if $\Th$ is compact, the proposed algorithm exhibits good error control properties. In order to get  a rough idea about the  proposal,  we propose to perform finite adaptation with $J$ adaptation points during the burn-in period.  Algorithm~\ref{alg:ABSL} outlines the proposed Approximated BSL  (ABSL) method.
For the simulations we report on in the next section,  we have used $c=1.5$ and $J=15$ to be consistent with AABC-MCMC, ABC-MCMC-M and ABC-SMC procedures.
\section{Simulations}
\label{sec:sim}
We analyze the following statistical models:
\begin{enumerate}
\item[(MA2)] Simple Moving Average model of lag 2;
\item[(R)] Ricker's model;
\item[(SVG)] Stochastic volatility with Gaussian emission noise;
\item[(SVS)] Stochastic volatility with $\al$-Stable errors.
\end{enumerate}
For all these models, the simulation of pseudo data for any parameter is simple and computationally fast, but the use of standard estimation methods can be quite challenging, especially for (R), (SVG) and (SVS). 
For ABC samplers before running a MCMC chain we estimate initial and final thresholds $\eps_0$ and $\eps_{15}$ (15 equal steps in log scale were used for all models) and the matrix $A$ which is used to calculate the discrepancy  $\delta = d(S(\by),S(\by_0))=(S(\by)-S(\by_0)^T A (S(\by)-S(\by_0))$. \\
To estimate $A$, we use the following steps:
\begin{itemize}
\item Set $A=\mathbf{I}_d$
\item Repeat steps I and II below for $J$ times ($J$=3 in our implementations)
\begin{description}
\item[I] Generate 500 pairs $\{\z_i,\by_i\}_{i=1}^{500}$ from $p(\z)f(\by|\z)$ and calculate discrepancies
$\{\z_i,\delta_i\}_{i=1}^{500}$ with $\delta_i = d(S(\by_i),S(\by_0))$
\item[II] Let $\z^*$ with smallest discrepancy. Finally generate 100 pseudo-data $(\by_1,\ldots,\by_{100})$ from $f(\by|\z^*)$, compute corresponding summary statistics $(s_1,\ldots,s_{100})$ and set $A$ to be the inverse of covariance matrix of $(s_1,\ldots,s_{100})$. 
\end{description}
\end{itemize}

We set $\eps_0$ to be the 5\% quantile of the observed discrepancies. The final $\eps_{15}$ is obtained by implementing a Random Walk version of Algorithm~\ref{alg:ABC-MCMC-M} and decreasing $\eps_0$  gradually by setting $\eps_j$ as  the  1\% quantile of  discrepancies $\delta$  corresponding to accepted samples generated  between adaption points $a_{j-1}$ and $a_j$, for $2\le j \le 15$.

The number of simulations was set to 500 and 100 just for computational convenience and is not driven by any theoretical arguments. 

We compare the following algorithms:
\begin{enumerate}
\item[(SMC)] Standard Sequential Monte Carlo for ABC;
\item[(ABC-RW)] The modified ABC-MCMC algorithm which updates $\eps$ and the random walk Metropolis transition kernel during burn-in; 
\item[(ABC-IS)]  The modified ABC-MCMC algorithm which updates $\eps$ and the Independent Metropolis transition kernel during burn-in;
\item[(BSL-RW)] Modified BSL where it adapts the random walk Metropolis transition kernel during burn-in;  
\item[(BSL-IS)] Modified BSL where it adapts the independent Metropolis transition kernel during burn-in; 
\item[(AABC-U)] Approximated ABC-MCMC with independent proposals and  uniform (U) weights;
\item[(AABC-L)] Approximated ABC-MCMC with independent proposals and  linear (L) weights;
\item[(ABSL-U)] Approximated BSL-MCMC with independent proposals and  uniform (U) weights;
\item[(AABC-L)] Approximated BSL-MCMC with independent proposals and  linear (L) weights.
\item[(Exact)]  Likelihood is computable and posterior samples are generated using an MCMC algorithm that is example-specific. 
\end{enumerate}

For SMC 500 particles were used, total number of iterations for ABC-RW, ABC-IS, AABC-U, AABC-L, ABSL-U and ABSL-L is 50000 with 10000 for burn-in. Since BSL-RW and BSL-IS are much more computationally expensive, total number of iterations were fixed at 10000 with 2000 burn-in and 50 pseudo-data simulations for every proposed parameter value (i.e. $m=50$). The Exact chain was run for 5000 iterations and 2000 for burn-in. It must be pointed out that all approximate samplers are based on the same summary statistics, same discrepancy function and the same $\eps$ sequence, so that they all start with the same initial conditions. 

For more reliable results we compare these sampling algorithms under data set replications. In this study we set the number of replicates $R= 100$, so that for each model 100 data sets were generated and each one was analyzed with the described above sampling methods. Various statistics and measures of efficiency were calculated for every model and data set, letting $\th_{rs}^{(t)}$ represent posterior samples from replicate $r=1,\cdots,R$, iteration $t=1,\cdots,M$ and parameter component $s=1,\cdots,q$ and similarly $\tilde \th_{rs}^{(t)}$ posterior from an exact chain (all draws are after burn-in period). We let $\th^{true}_s$ denote the true parameter that generated the data. Moreover let $D_{rs}(x)$, $\tilde D_{rs}(x)$ be estimated density function at replicate $r=1,\cdots,R$ and components $s=1,\cdots,q$ for approximate and exact chains respectively. Then the following quantities are defined:
\begin{equation*}
\begin{split}
& \mbox{Diff in mean (DIM)} = Mean_{r,s}(|Mean_t(\th_{rs}^{(t)}) - Mean_t(\tilde\th_{rs}^{(t)})|), \\
& \mbox{Diff in covariance (DIC)} = Mean_{r,s}(|Cov_t(\th_{rs}^{(t)}) - Cov_t(\tilde\th_{rs}^{(t)})|), \\
& \mbox{Total Variation (TV)} = Mean_{r,s}\left(0.5\int |D_{rs}(x)-\tilde D_{rs}(x)| dx \right), \\
& \mbox{Bias}^2 = Mean_s\left(\left( Mean_{tr}(\th_{rs}^{(t)}) - \th^{true}_s \right)^2\right), \\
& \mbox{VAR} = Mean_s(Var_r(Mean_t(\th_{rs}^{(t)}))), \\
& \mbox{MSE} = \mbox{Bias}^2 + \mbox{VAR},
\end{split}
\end{equation*}
where $Mean_t(a_{st})$ is defined as average of $\{a_{st}\}$ over index $t$ and in similar manner $Var_t(a_{st})$ and $Cov_t(a_{st})$ representing variance and covariance respectively. The first three measures are useful in determining how close posterior draws from different samplers are to the draws generated by the exact chain (when it is available). On the other hand the last three are standard quantities that measure how close in mean square posterior means are to the true parameters that generated the data. To study efficiency of proposed algorithms we need to take into account CPU time that it takes to run a chain as well as auto-correlation properties. Define auto-correlation time (ACT) for every parameter's component and replicate of samples $\th_{rs}^{(t)}$ as:
$$ \mbox{ACT}_{rs} = 1 + 2\sum_{a =1}^{\infty}\rho_a(\th_{rs}^{(t)}), $$
where $\rho_a$ is auto-correlation coefficient at lag $a$. In practice we sum all the lags up to the first negative correlation. Letting $M-B$ to be number of chain iterations (after burn-in) and $CPU_{r}$ correspond to total CPU time to run the whole chain during replicate $r$, we use Effective Sample Size (ESS) and Effective Sample Size per CPU (ESS/CPU) as:
\begin{equation}
\begin{split}
& \mbox{ESS} = Mean_{rs}((M-B)/ \mbox{ACT}_{rs}), \\
& \mbox{ESS/CPU} = Mean_{rs}((M-B)/ \mbox{ACT}_{rs}/CPU_r). 
\end{split}
\end{equation} 
Note that these indicators are averaged over parameter components and replicates.
ESS intuitively can be thought as approximate number of "independent" samples out of $M-B$, the higher is ESS the more efficient is the sampling algorithm, when ESS is combined with CPU (ESS/CPU) it provides a powerful indicator for MCMC's efficiency. Generally a sampler with highest ESS/CPU  is preferred as it produces larger number of "independent" draws per unit time. 
\subsection{Moving Average Model}
A popular toy example to check performances of ABC and BSL techniques is MA2 model:
\begin{equation}
\begin{split}
&z_{i}\overset{iid}{\sim}\calN(0,1); \hsp i=\{-1,0,1,\cdots,n\}, \\
& y_i = z_i + \theta_1 z_{i-1} + \theta_2 z_{i-2}; \hsp i=\{1,\cdots,n\}.
\end{split}
\end{equation}
The data are represented by the sequence $\by=\{y_1,\cdots,y_n\}$. It is well known that $Y_i$ follow a stationary distribution for any $\th_1,\th_2$, but there are conditions required for identifiability. Hence, we impose uniform prior on the following set:
\begin{equation}
\begin{split}
\th_1+\th_2 > -1,& \\
\th_1-\th_2 < 1,&  \\
-2 < \th_1 < 2,& \\
-1 < \th_2 < 2.& 
\end{split}
\end{equation}  
It is very easy to see that the joint distribution of $\by$ is multivariate Gaussian with mean 0, diagonal variances $1+ \th_1^2 + \th_2^2$, covariance at lags 1 and 2, $\th_1+\th_1\th_2$ and $\th_2$ respectively and zero at other lags. In this case, (Exact) sampling is feasible. 
For simulations we set $\{\th_1=0.6,\th_2=0.6\}$, $n=200$ and define summary statistics $S(\by)=(\hat \gamma_0(\by),\hat \gamma_1(\by),\hat \gamma_2(\by))$ as sample variance and covariances at lags 1 and 2.  
First we show results based on one replicate. Figure~\ref{fig:ma-trace-wabcu} shows the trace plots, histograms and auto-correlation functions estimated from posterior draws for parameters $\theta_1$ and $\theta_2$ for the AABC-U sampler. Note that only post burn-in samples are shown.
\begin{figure}[!ht]
\begin{center}
\caption {MA2 model: AABC-U Sampler. Each row corresponds to parameters $\theta_1$ (top row) and $\theta_2$ (bottom row)  and shows in order from left to right: Trace-plot, Histogram and Auto-correlation function. Red lines represent true parameter values. }
\includegraphics[scale=0.40]{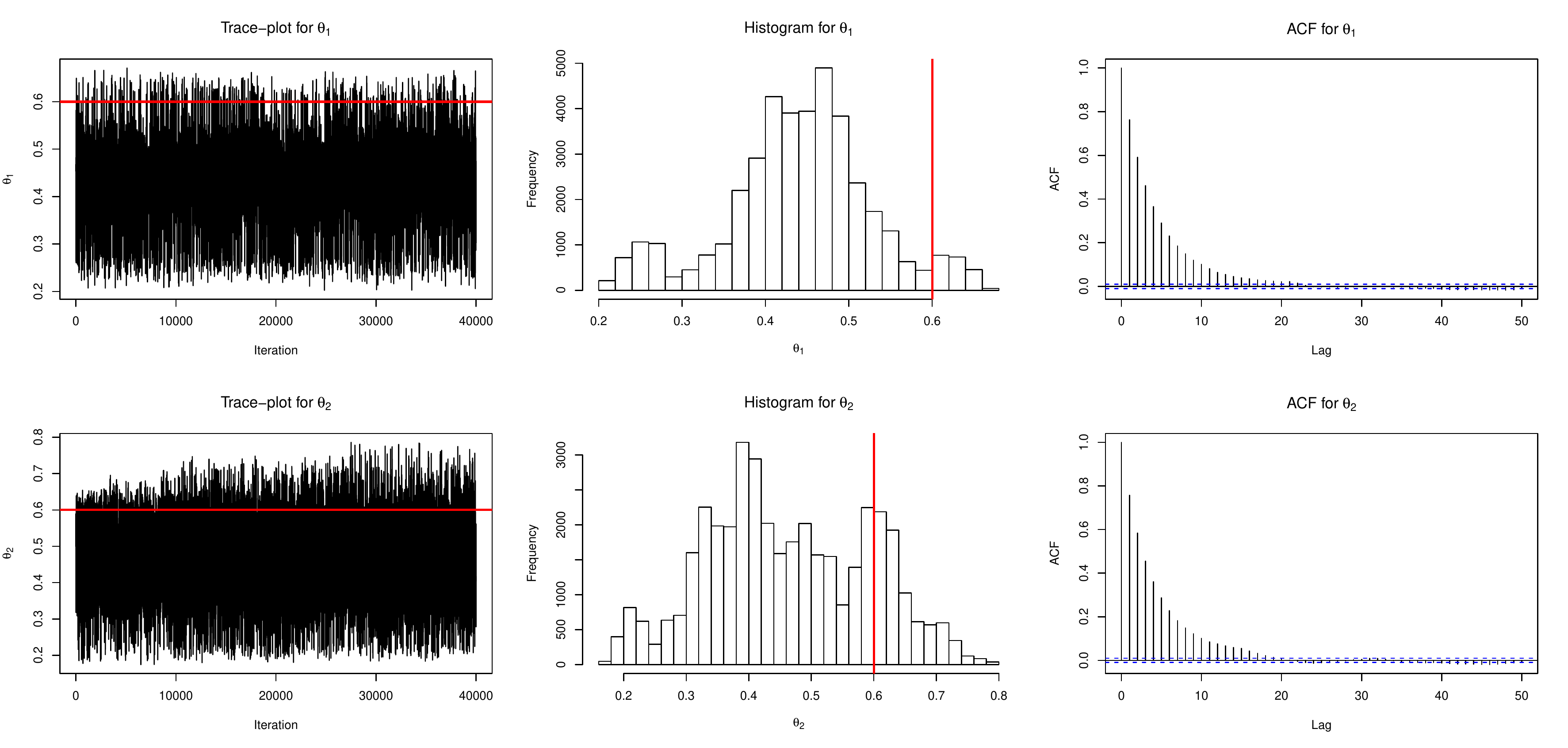}
\label{fig:ma-trace-wabcu}
\end{center}
\end{figure}
\begin{figure}[!ht]
\begin{center}
\caption {MA2 model: ABSL-U Sampler. Each row corresponds to parameters $\theta_1$ (top row) and $\theta_2$ (bottom row)  and shows in order from left to right: Trace-plot, Histogram and Auto-correlation function. Red lines represent true parameter values. }
\includegraphics[scale=0.40]{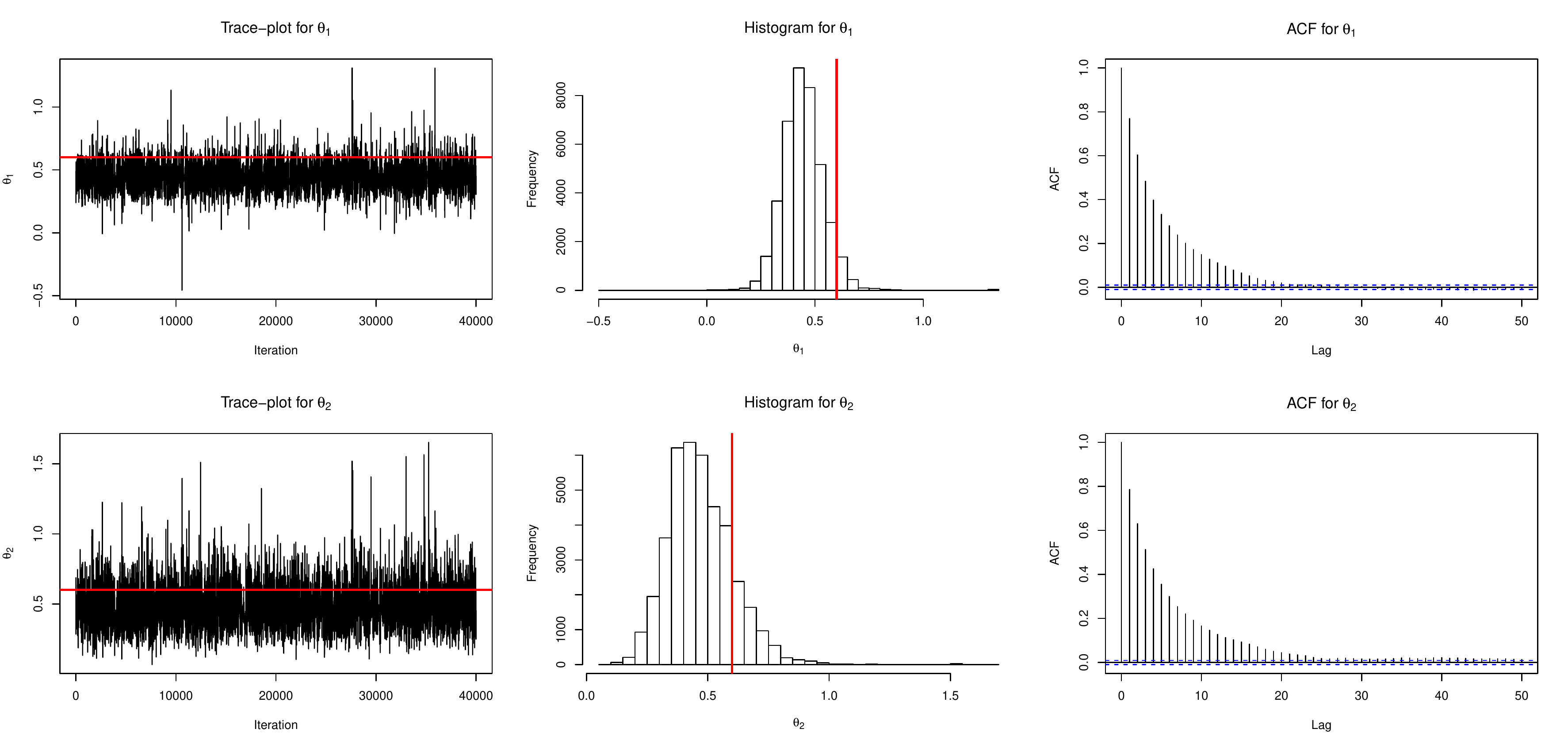}
\label{fig:ma-trace-wbslu}
\end{center}
\end{figure}
\begin{figure}[!ht]
\begin{center}
\caption {MA2 model: ABC-RW Sampler. Each row corresponds to parameters $\theta_1$ (top row) and $\theta_2$ (bottom row)  and shows in order from left to right: Trace-plot, Histogram and Auto-correlation function. Red lines represent true parameter values. }
\includegraphics[scale=0.40]{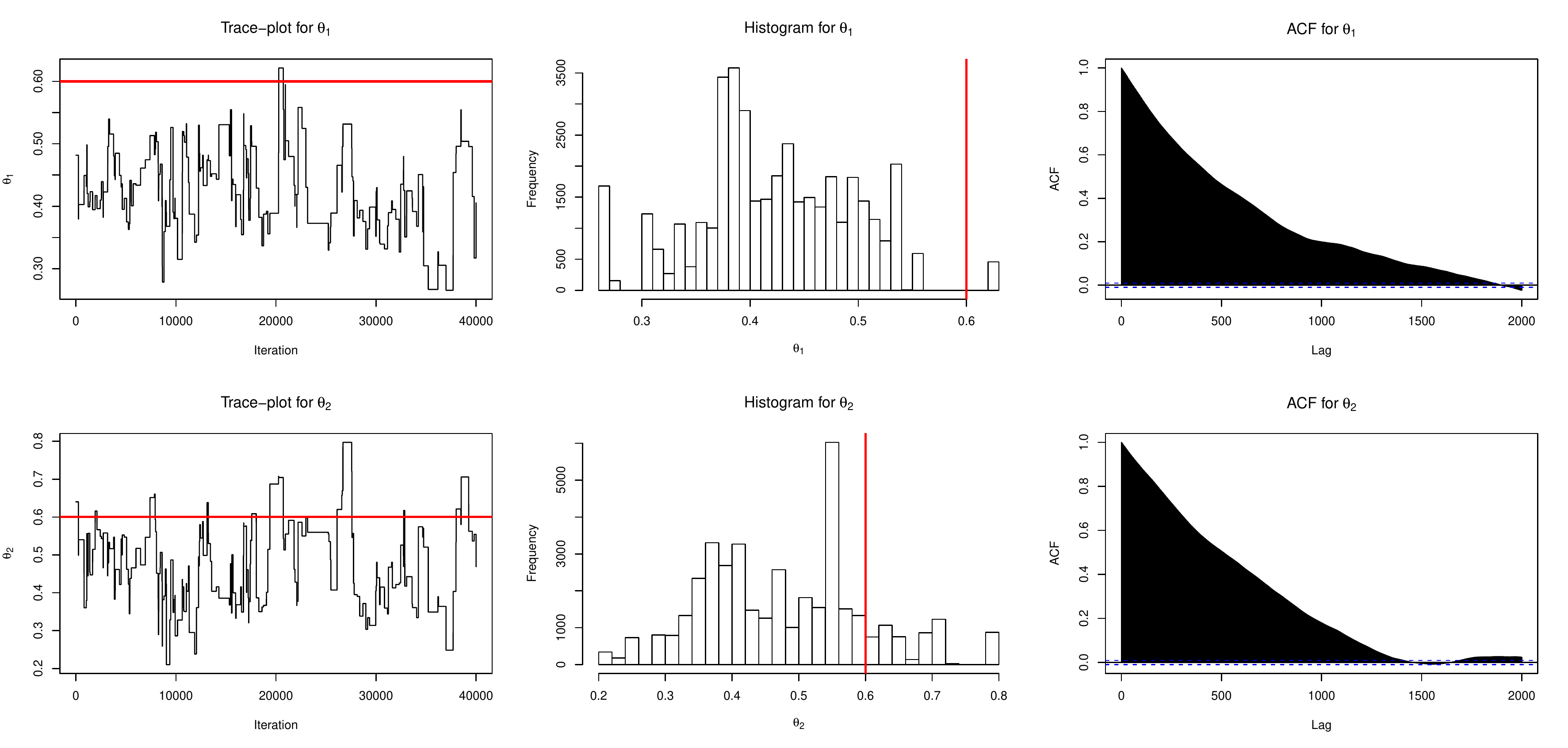}
\label{fig:ma-trace-abcrw}
\end{center}
\end{figure}
Similarly, Figure~\ref{fig:ma-trace-wbslu} and  Figure~\ref{fig:ma-trace-abcrw} display the behaviour of ABSL-U sampler and standard ABC-RW, respectively. From these plots it is apparent that the proposed AABC-U and ABSL-U have much better mixing than ABC-RW.   In the interest of keeping the paper length within reasonable limits, we briefly mention that additional simulations suggest that  AABC-L is similar to AABC-U and ABSL-L to ABSL-U, while  ABC-IS is  outperformed by ABC-RW. 

In order to summarize and compare the information in the MCMC draws produced by the approximated samplers and  the exact chain, we plot the estimated densities in Figure~\ref{fig:ma-density-BSLIS}. The left and right side plots refer to $\theta_1$ and $\theta_2$, respectively.  The two upper plots  compare the estimated density of the exact MCMC sampler with ABC-based ones (SMC, ABC-RW and AABC-U), while the two lower plots compare the exact sampler with Synthetic Likelihood based methods (BSL-IS and ABSL-U).
\begin{figure}[!ht]
\begin{center}
\caption {MA model: Estimated densities for each component. First row compares Exact, SMC, ABC-RW and AABC-U samplers. Second row compares Exact, BSL-IS and ABSL-U. Columns correspond to parameter's components, from left to right: $\theta_1$ and $\theta_2$.}
\includegraphics[scale=0.40]{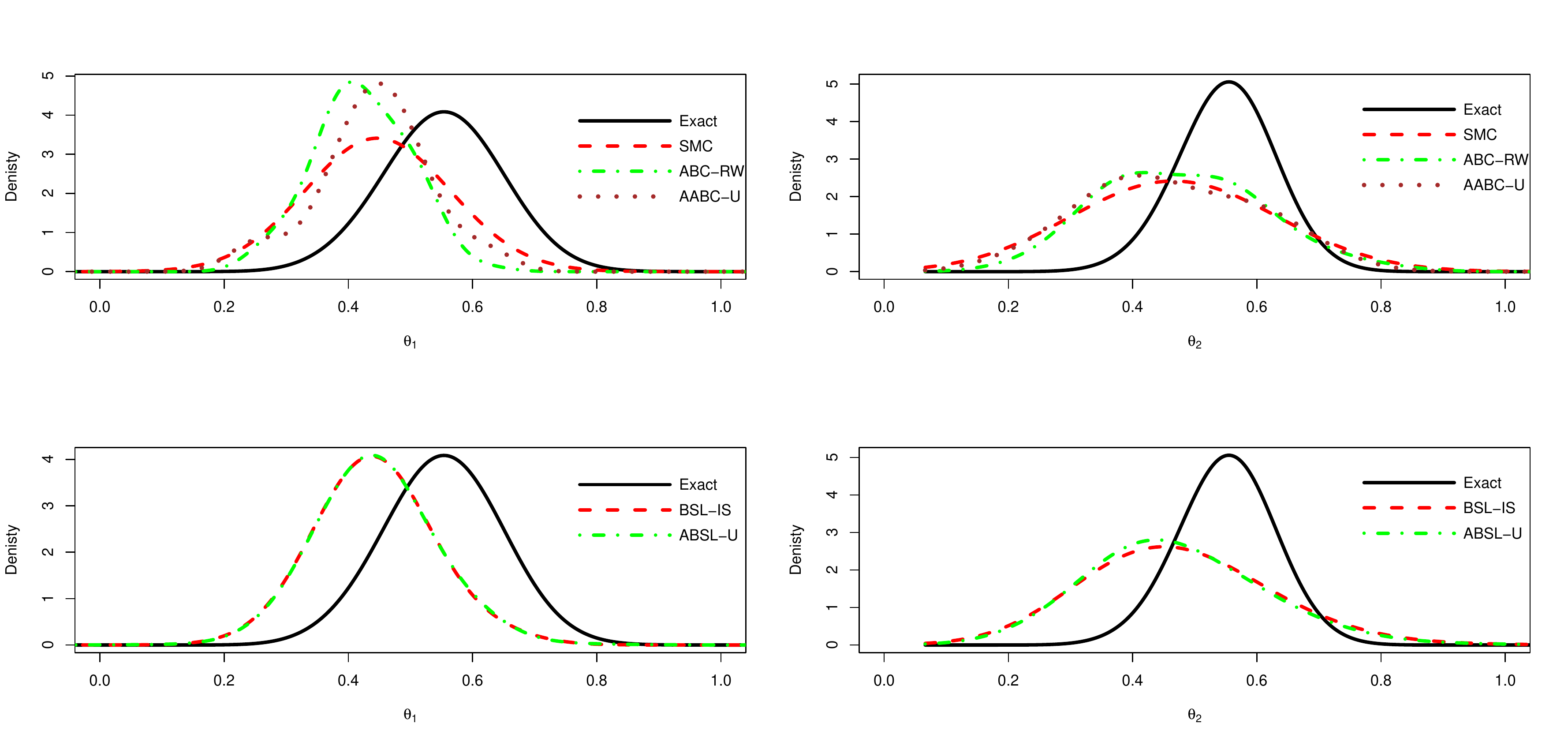}
\label{fig:ma-density-BSLIS}
\end{center}
\end{figure}

The posterior distributions evaluated from  AABC-U is very similar to those produced by SMC and ABC-RW, but all are distinct from the Exact one. This latter difference may be due to the loss of information incurred when the posterior is conditional on a non-sufficient statistic. Similarly, the distribution produced by ABSL-U draws is very close to that of BSL-IS. These observations hold  for both components, $\th_1$ and $\th_2$. \\
To study accuracy, precision and efficiency of proposed samplers we perform a simulation study where 100 data sets are generated and all samplers are run for every data set. The results are summarized in Table~\ref{table:ma}.
\begin{table} [!ht]
\begin{center}
\caption{ Simulation Results (MA model): Average Difference in mean, Difference in covariance, Total variation, square roots of Bias, Variance and MSE, Effective sample size and Effective sample size per CPU time, for every sampling algorithm. }
\smallskip
\scalebox{1.0}{
\begin{tabular}{l | l l l || l l l || l l  }
\multicolumn{1}{c}{ }  &  \multicolumn{3}{c}{Diff with exact}   & \multicolumn{3}{c}{Diff with true parmater} & \multicolumn{2}{c}{Efficiency}\\
\hline
Sampler         & DIM & DIC & TV & $\sqrt{\mbox{Bias}^2}$ & $\sqrt{\mbox{VAR}}$ & $\sqrt{\mbox{MSE}}$ & ESS & ESS/CPU\\
\hline
SMC	&	0.082	&	0.0045	&	0.418	&	0.014	&	0.115	&	0.116	&	471	&	0.505\\
ABC-RW	&	0.088	&	0.0063	&	0.466	&	0.016	&	0.123	&	0.124	&	23	&	0.231	\\
ABC-IS	&	0.084	&	0.0067	&	0.455	&	0.016	&	0.115	&	0.116	&	44	&	0.389	\\
AABC-U	&	0.083	&	0.0071	&	0.444	&	0.018	&	0.116	&	0.117	&	3446	&	6.215 \\
AABC-L	&	0.080	&	0.0067	&	0.438	&	0.017	&	0.112	&	0.113	&	2820	&	5.107	\\
BSL-RW	&	0.082	&	0.0070	&	0.438	&	0.015	&	0.114	&	0.115	&	252	&	0.282\\
BSL-IS	&	0.081	&	0.0070	&	0.436	&	0.015	&	0.114	&	0.115	&	841	&	0.923\\
ABSL-U	&	0.081	&	0.0095	&	0.443	&	0.017	&	0.114	&	0.115	&	3950	&	5.584	\\
ABSL-L	&	0.082	&	0.0078	&	0.441	&	0.015	&	0.114	&	0.115	&	4165	&	6.030	\\
\hline
\end{tabular}
}
\label{table:ma}
\end{center}
\end{table}
Examining this table we immediately note that ESS/CPU measure is much larger for proposed algorithms than for standard methods. The improvement is very substantial, for example ESS/CPU for AABC-U is 12 times larger than for the best standard ABC procedures like SMC. Similar results are shown for Bayesian Synthetic Likelihood. We also examine DIM, DIC, TV and MSE quantities that provide information about the proximity of approximate samples to the exact MCMC ones. For all these quantities the smaller the value the better is the sampler.  We see that all these measures for AABC-U and AABC-L are very similar to SMC, ABC-RW and ABC-IS and frequently outperforms them. Similarly for BSL approach. Another observation is that the approximated algorithm with uniform and linear weights generally perform very similarly.    

\subsection{Ricker's Model}
Ricker's model is analyzed very frequently to test Synthetic Likelihood procedures \cite{wood2010statistical,price2018bayesian}. It is a particular instance of hidden Markov model:
\begin{equation}
\begin{split}
&x_{-49}=1;\hsp z_i\overset{iid}{\sim}\calN(0,\exp(\th_2)^2); \hsp i=\{-48,\cdots,n\}, \\
&x_{i} = \exp(\exp(\th_1))x_{i-1}\exp(-x_{i-1}+z_i); \hsp i=\{-48,\cdots,n\}, \\
&y_{i} = Pois(\exp(\th_3)x_i);  \hsp i=\{-48,\cdots,n\},
\end{split}
\end{equation}
where $Pois(\lambda)$ is Poisson distribution with mean parameter $\lambda$ and $n=100$. Only $\by=(y_1,\cdots,y_n)$ sequence is observed, because the first 50 values are ignored. Note that all parameters $\th=(\th_1,\th_2,\th_3)$ are unrestricted, the prior is given as (each prior parameter is independent):
\begin{equation}
\begin{split}
& \th_1 \sim \calN(0,1), \\
& \th_2 \sim Unif(-2.3,0), \\
& \th_3 \sim \calN(0,4).
\end{split}
\end{equation} 
We restrict the range of $\th_2$ as all algorithms become unstable for $\th_2$ outside this interval. Note that the marginal distribution of $\by$ is not available in closed form, but transition distribution of hidden variables $X_i|x_{i-1}$ and emission probabilities $Y_i|x_i$ are known and hence we can run Particle MCMC (PMCMC) \cite{andrieu2010particle} or Ensemble MCMC \cite{shestopaloff2013mcmc} to sample from the posterior distribution $\pi(\th|\by_0)$. {Here we are utilizing the Particle MCMC with 100 particles.}
As suggested in \cite{wood2010statistical} we set $\th_0=(\log(3.8),0.9,2.3)$ and define summary statistics $S(\by)$ as the 14-dimensional vector whose components are:
\begin{enumerate}
\item[(C1)] \#$\{i: y_i = 0\}$,
\item[(C2)] Average of $\by$, $\bar y$,
\item[(C3:C7)] Sample auto-correlations at lags 1 through 5,
\item[(C8:C11)] Coefficients $\beta_0,\beta_1,\beta_2,\beta_3$ of cubic regression\\
 $(y_i-y_{i-1}) = \beta_0 + \beta_1 y_i + \beta_2 y_i^2 + \beta_3 y_i^3 + \eps_i$, $i=2,\ldots,n$,
\item[(C12-C14)] Coefficients $\beta_0,\beta_1,\beta_2$ of quadratic regression \\
 $y_i^{0.3} = \beta_0 + \beta_1 y_{i-1}^{0.3} + \beta_2 y_{i-1}^{0.6} + \eps_i$, $i=2,\ldots,n$.
\end{enumerate}

Figures~\ref{fig:ricker-trace-wabcu}, \ref{fig:ricker-trace-wbslu} and \ref{fig:ricker-trace-abcrw} show trace-plots, histograms and ACF function for AABC-U, ABSL-U and ABC-RW samplers for each component (red lines correspond to the true parameter).
\begin{figure}[!ht]
\begin{center}
\caption {Ricker's model: AABC-U Sampler. Each row corresponds to parameters $\theta_1$ (top row), $\theta_2$ (middle row) and $\theta_3$ (bottom row)  and shows in order from left to right: Trace-plot, Histogram and Auto-correlation function. Red lines represent true parameter values. }
\includegraphics[scale=0.40]{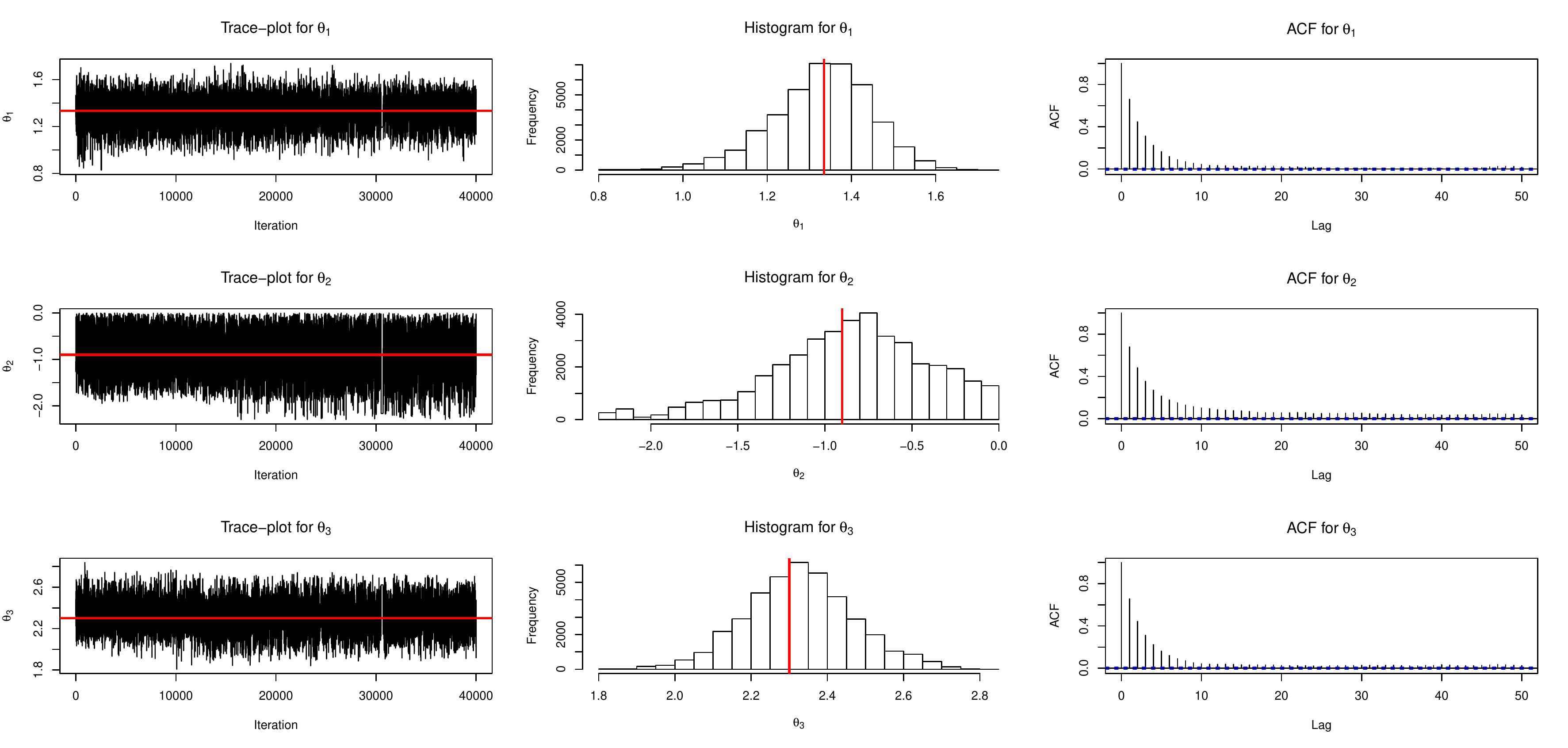}
\label{fig:ricker-trace-wabcu}
\end{center}
\end{figure}
\begin{figure}[!ht]
\begin{center}
\caption {Ricker's model: ABSL-U Sampler. Each row corresponds to parameters $\theta_1$ (top row), $\theta_2$ (middle row) and $\theta_3$ (bottom row)  and shows in order from left to right: Trace-plot, Histogram and Auto-correlation function. Red lines represent true parameter values. }
\includegraphics[scale=0.40]{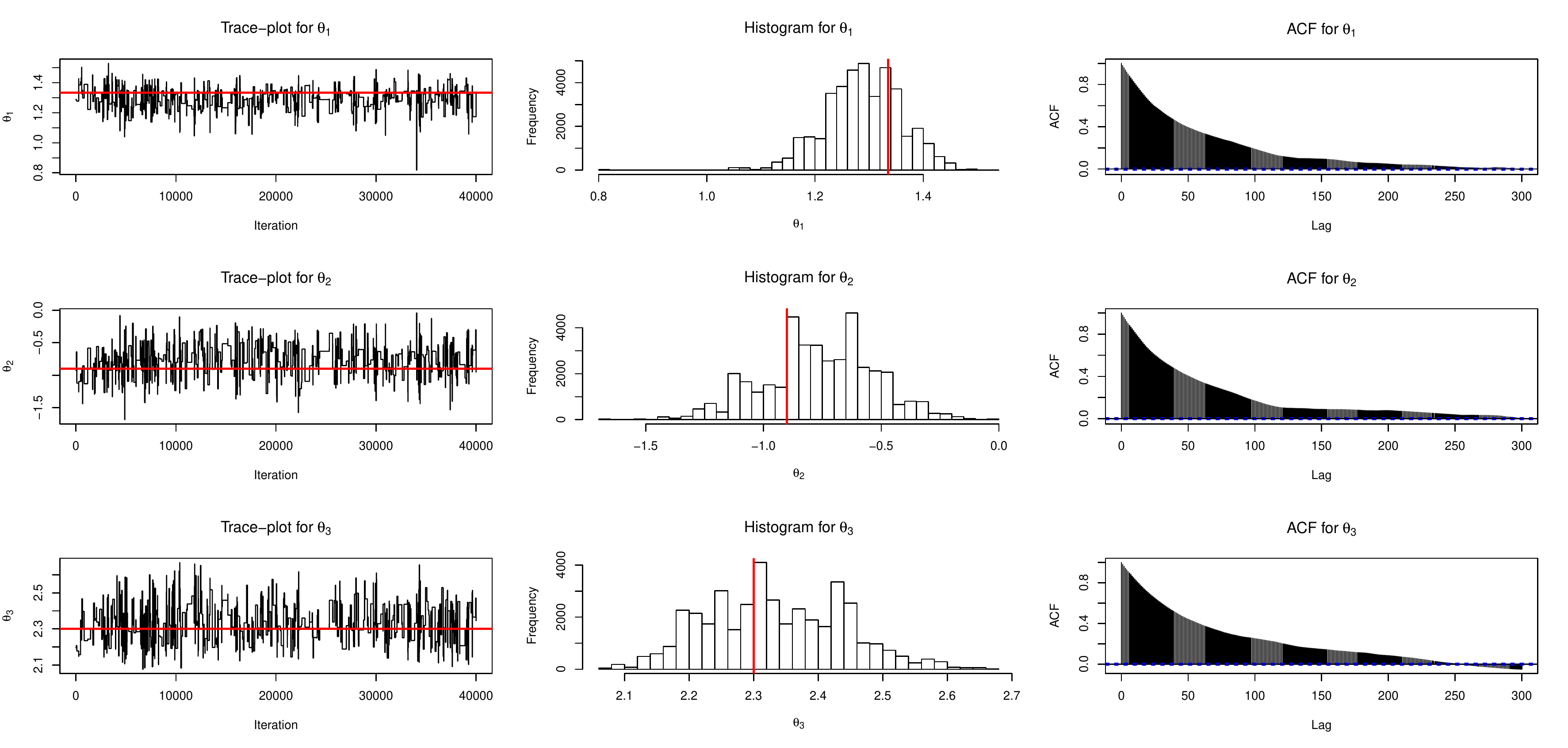}
\label{fig:ricker-trace-wbslu}
\end{center}
\end{figure}
\begin{figure}[!ht]
\begin{center}
\caption {Ricker's model: ABC-RW Sampler. Each row corresponds to parameters $\theta_1$ (top row), $\theta_2$ (middle row) and $\theta_3$ (bottom row)  and shows in order from left to right: Trace-plot, Histogram and Auto-correlation function. Red lines represent true parameter values.}
\includegraphics[scale=0.40]{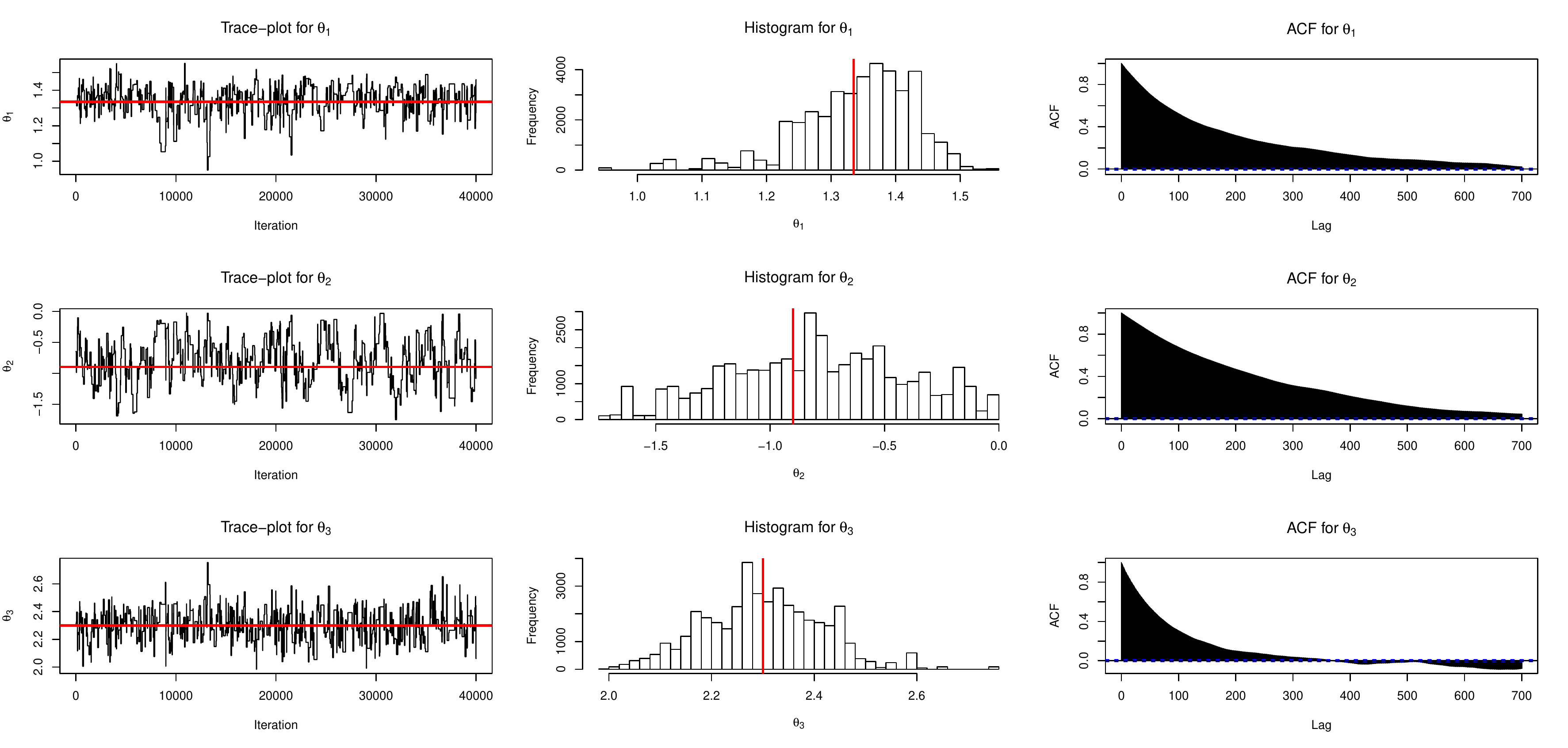}
\label{fig:ricker-trace-abcrw}
\end{center}
\end{figure}
We show here ABC-RW instead of ABC-IS because the latter exhibits a poorer performance.  The main observation is that mixing of AABC-U is much better than in ABC-RW with smaller auto-correlation values.  ABSL-U has higher auto-correlations than AABC-U but still performs quite well.
To see how close the draws from simulation-based algorithms to the draws from the Exact chain, we plot the estimated approximate  posterior marginal densities in Figure~\ref{fig:ricker-density-BSLRW}. The two upper plots (left and right are associated to parameter's component) compares estimated density of exact PMCMC sampler (with 100 particles) with ABC-based ones (SMC, ABC-RW and AABC-U), two lower plots compare the Exact sampler with Synthetic Likelihood based methods (BSL-RW and ABSL-U).

\begin{figure}[!ht]
\begin{center}
\caption {Ricker's model: Estimated posterior marginal densities for each component. First row compares Exact, SMC, ABC-RW and AABC-U samplers. Second row compares Exact, BSL-RW and ABSL-U. Columns correspond to parameter's components, from left to right: $\theta_1$, $\theta_2$ and $\theta_3$.}
\includegraphics[scale=0.40]{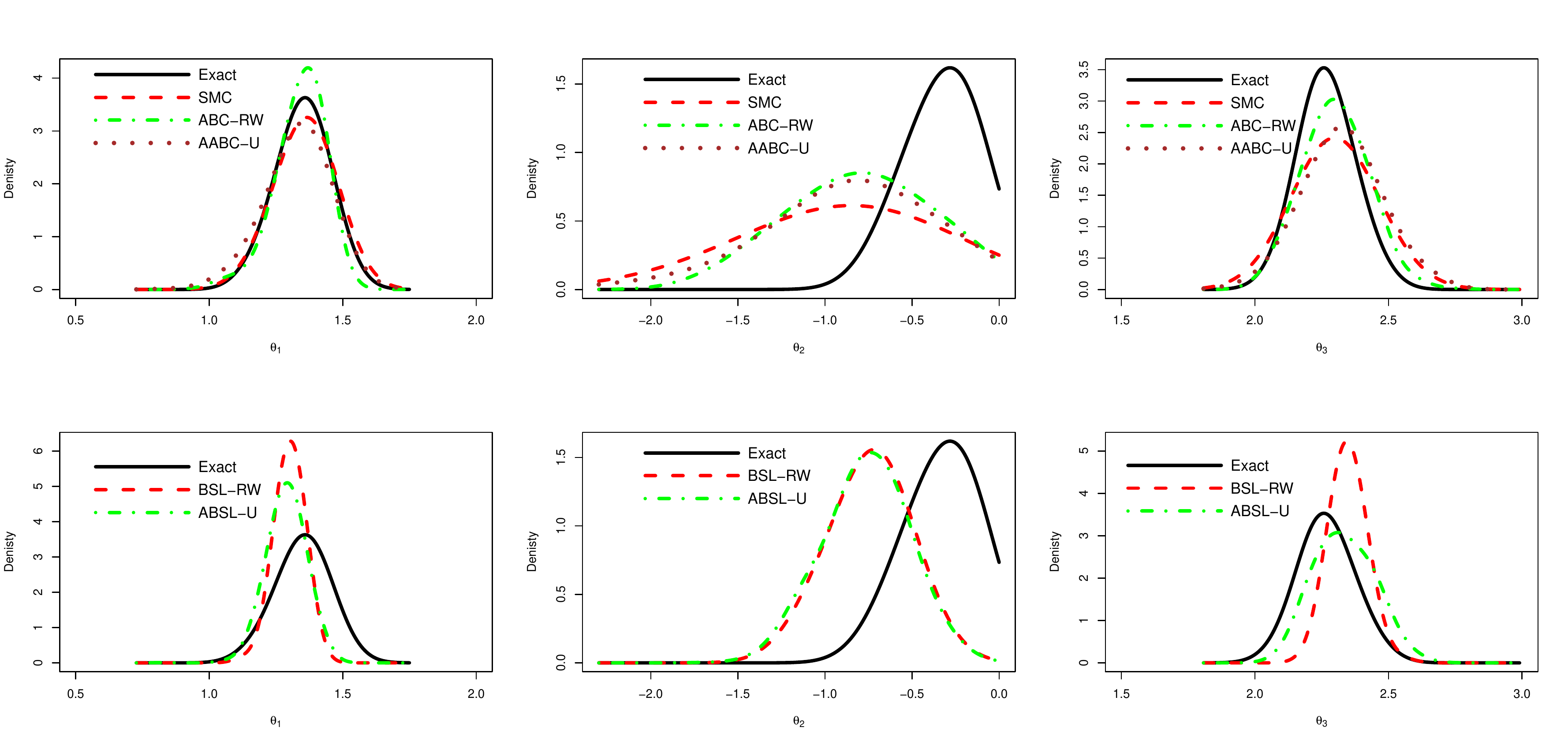}
\label{fig:ricker-density-BSLRW}
\end{center}
\end{figure}
Note that ABC-based samplers (SMC, ABC-RW and AABC-U) have very similar estimated densities. The densities of Synthetic Likelihood methods are also similar. For the second component there is a large difference between exact and approximate posteriors which may be caused by the loss of information induced by the choice of summary statistics.

A more general study, where results are averaged over 100 independent replicates, is shown in Table~\ref{table:ricker}.
\begin{table} [!ht]
\begin{center}
\caption{Simulation Results (Ricker's model): Average Difference in mean, Difference in covariance, Total variation, square roots of Bias, variance and MSE, Effective sample size and Effective sample size per CPU time, for every sampling algorithm. }
\smallskip
\scalebox{1.0}{
\begin{tabular}{l | l l l || l l l || l l  }
\multicolumn{1}{c}{ }  &  \multicolumn{3}{c}{Diff with exact}   & \multicolumn{3}{c}{Diff with true parmater} & \multicolumn{2}{c}{Efficiency} \\
\hline
Sampler         & DIM & DIC & TV & $\sqrt{\mbox{Bias}^2}$ &  $\sqrt{\mbox{VAR}}$ & $\sqrt{\mbox{MSE}}$ & ESS & ESS/CPU\\
\hline
SMC	&	0.152	&	0.0177	&	0.378	&	0.086	&	0.201	&	0.219	&	472	&	0.521\\
ABC-RW	&	0.135	&	0.0201	&	0.389	&	0.059	&	0.180	&	0.189	&	87	&	0.199	\\
ABC-IS	&	0.139	&	0.0215	&	0.485	&	0.063	&	0.195	&	0.205	&	47	&	0.099	\\
AABC-U	&	0.147	&	0.0279	&	0.402	&	0.076	&	0.190	&	0.204	&	3563	&	4.390	\\
AABC-L	&	0.141	&	0.0258	&	0.392	&	0.070	&	0.189	&	0.201	&	4206	&	5.193	\\
BSL-RW	&	0.129	&	0.0080	&	0.382	&	0.038	&	0.206	&	0.209	&	131	&	0.030\\
BSL-IS	&	0.122	&	0.0082	&	0.455	&	0.022	&	0.197	&	0.198	&	33	&	0.007\\
ABSL-U	&	0.103	&	0.0054	&	0.377	&	0.023	&	0.170	&	0.171	&	284	&	0.180\\
ABSL-L	&	0.106	&	0.0051	&	0.382	&	0.012	&	0.173	&	0.173	&	207	&	0.135\\
\hline
\end{tabular}
}
\label{table:ricker}
\end{center}
\end{table}
Again, the proposed strategies clearly outperform in terms of overall efficiency (ESS/CPU). For instance, AABC-U is about 10 times more efficient than standard SMC and ABSL-U is 6 times more efficient than BSL-RW. At the same time DIM, DIC, TV and MSE are generally smaller for approximate methods than for standard ones. 

\subsection{Stochastic Volatility with Gaussian emissions}
When analyzing stationary time series,  it is frequently observed that there are periods of high and periods of low volatility. Such phenomenon is called \textit{volatility clustering}, see for example \citep{lux2000volatility}. One way to model such a behaviour is through a Stochastic Volatility (SV) model, where variances of the observed time series depend on hidden states that themselves form a stationary time series. Consider the following model which depends on three parameters $(\th_1,\th_2,\th_3)$:
\begin{equation}
\begin{split}
&x_{1} \sim \calN(0,1/(1-\th_1^2));\hsp v_i\overset{iid}{\sim}\calN(0,1);\hsp w_i\overset{iid}{\sim}\calN(0,1) ; \hsp i=\{1,\cdots,n\}, \\
&x_{i} = \th_1 x_{i-1} + v_i; \hsp i=\{2,\cdots,n\}, \\
&y_{i} = \sqrt{\exp(\th_2 + \exp(\th_3)x_i)}w_i;  \hsp i=\{1,\cdots,n\}.
\end{split}
\end{equation}
Only $\by=(y_1,\cdots,y_n)$ is observed while $(x_1,\cdots,x_n)$ are hidden states. The parameter $\th_1\in (-1, 1)$  controls the auto-correlation of hidden states, while $\th_2$ and $\th_3$ are unrestricted and relate to the hidden states influence on the variability of the observed series. Given a hidden state, the distribution of the observed variable is normal which may not be appropriate in some examples. We introduce the following priors, independently for each parameter:
\begin{equation}
\begin{split}
& \th_1 \sim Unif(0,1), \\
& \th_2 \sim \calN(0,1), \\
& \th_3 \sim \calN(0,1).
\end{split}
\end{equation} 
We set the true parameters to $(\th_1=0.95,\th_2=-2,\th_3=-1)$ and length of the time series $n=500$.  We use Particle MCMC (PMCMC)   as the Exact sampling scheme. Since pseudo-data sets can be easily generated for every parameter value, the SV is a good example to demonstrate the performances of the generative algorithms considered here. For summary statistics we use a 7-dimensional vector whose components are:
\begin{enumerate}
\item[(C1)] \#$\{i: y_i^2 > \mbox{quantile}(\by_0^2,0.99)\}$,
\item[(C2)] Average of $\by^2$,
\item[(C3)] Standard deviation of $\by^2$,
\item[(C4)] Sum of the first 5 auto-correlations of $\by^2$,
\item[(C5)] Sum of the first 5 auto-correlations of $\{\one_{\{y_i^2<\mbox{quantile}(\by^2,0.1)\}}\}_{i=1}^n$,
\item[(C6)] Sum of the first 5 auto-correlations of $\{\one_{\{y_i^2<\mbox{quantile}(\by^2,0.5)\}}\}_{i=1}^n$,
\item[(C7)] Sum of the first 5 auto-correlations of $\{\one_{\{y_i^2<\mbox{quantile}(\by^2,0.9)\}}\}_{i=1}^n$.
\end{enumerate}
Here $\mbox{quantile}(\by,\tau)$ is defined as $\tau$-quantile of the sequence $\by$. As was shown in \cite{schmitt2015quantile} and \cite{dette2015copulas} the auto-correlation of indicators (under different quantiles) can be very useful in characterizing a time series and that is why we have added (C5),(C6) and (C7) to the summary statistic. We focus here on $\by^2$ and its auto-correlations since model parameters only affect variability of $\by$ (auto-correlation of $\by$ is zero for any lag). 
Figures~\ref{fig:sv-trace-wabcu}, \ref{fig:sv-trace-wbslu} and \ref{fig:sv-trace-abcrw} show trace-plots, histograms and ACF function for AABC-U, ABSL-U and ABC-RW samplers respectively for each component (red lines correspond to the true parameter).
\begin{figure}[!ht]
\begin{center}
\caption {SV model: AABC-U Sampler. Each row corresponds to parameters $\theta_1$ (top row), $\theta_2$ (middle row) and $\theta_3$ (bottom row)  and shows in order from left to right: Trace-plot, Histogram and Auto-correlation function. Red lines represent true parameter values. 
}
\includegraphics[scale=0.40]{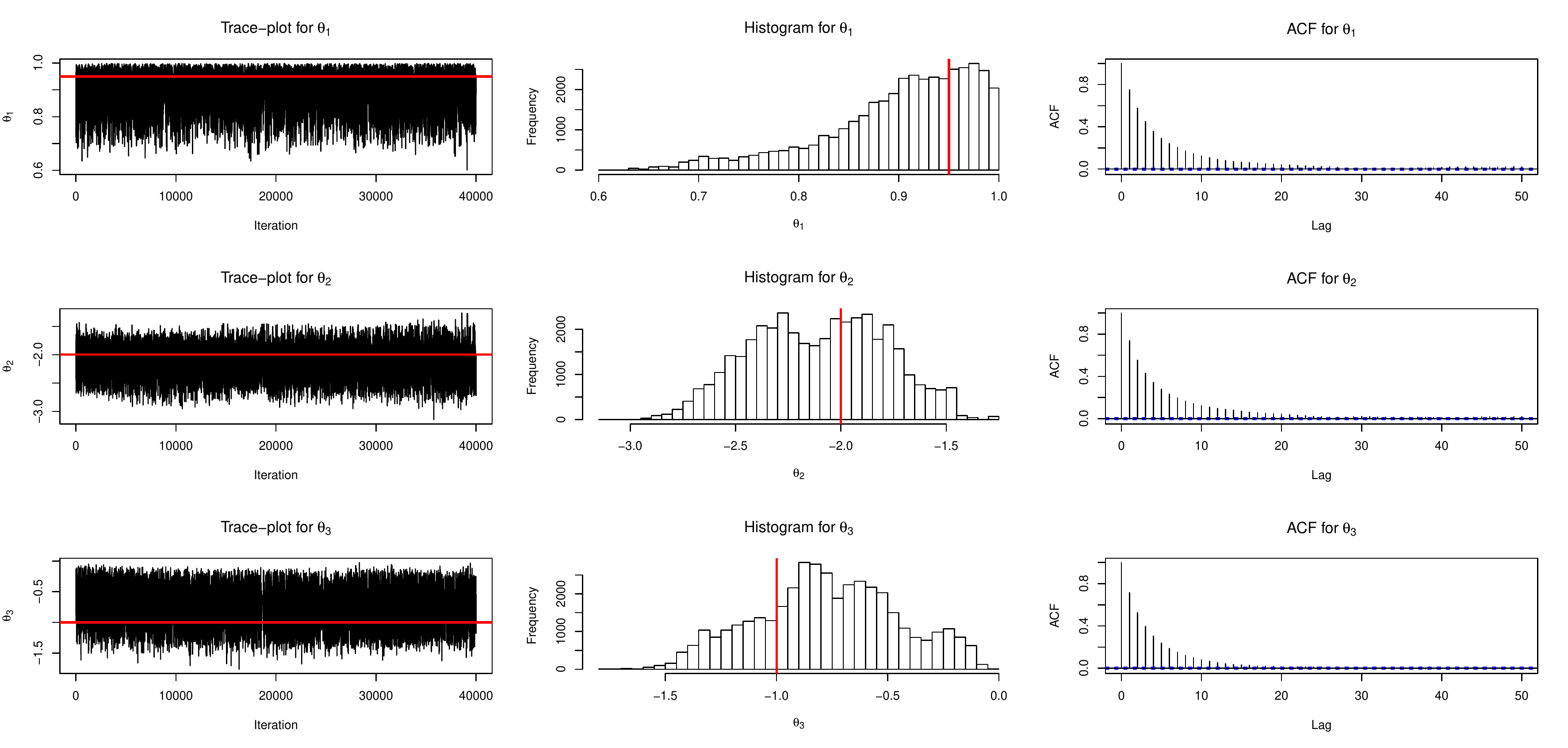}
\label{fig:sv-trace-wabcu}
\end{center}
\end{figure}
\begin{figure}[!ht]
\begin{center}
\caption {SV model: ABSL-U Sampler. Each row corresponds to parameters $\theta_1$ (top row), $\theta_2$ (middle row) and $\theta_3$ (bottom row)  and shows in order from left to right: Trace-plot, Histogram and Auto-correlation function. Red lines represent true parameter values. }
\includegraphics[scale=0.40]{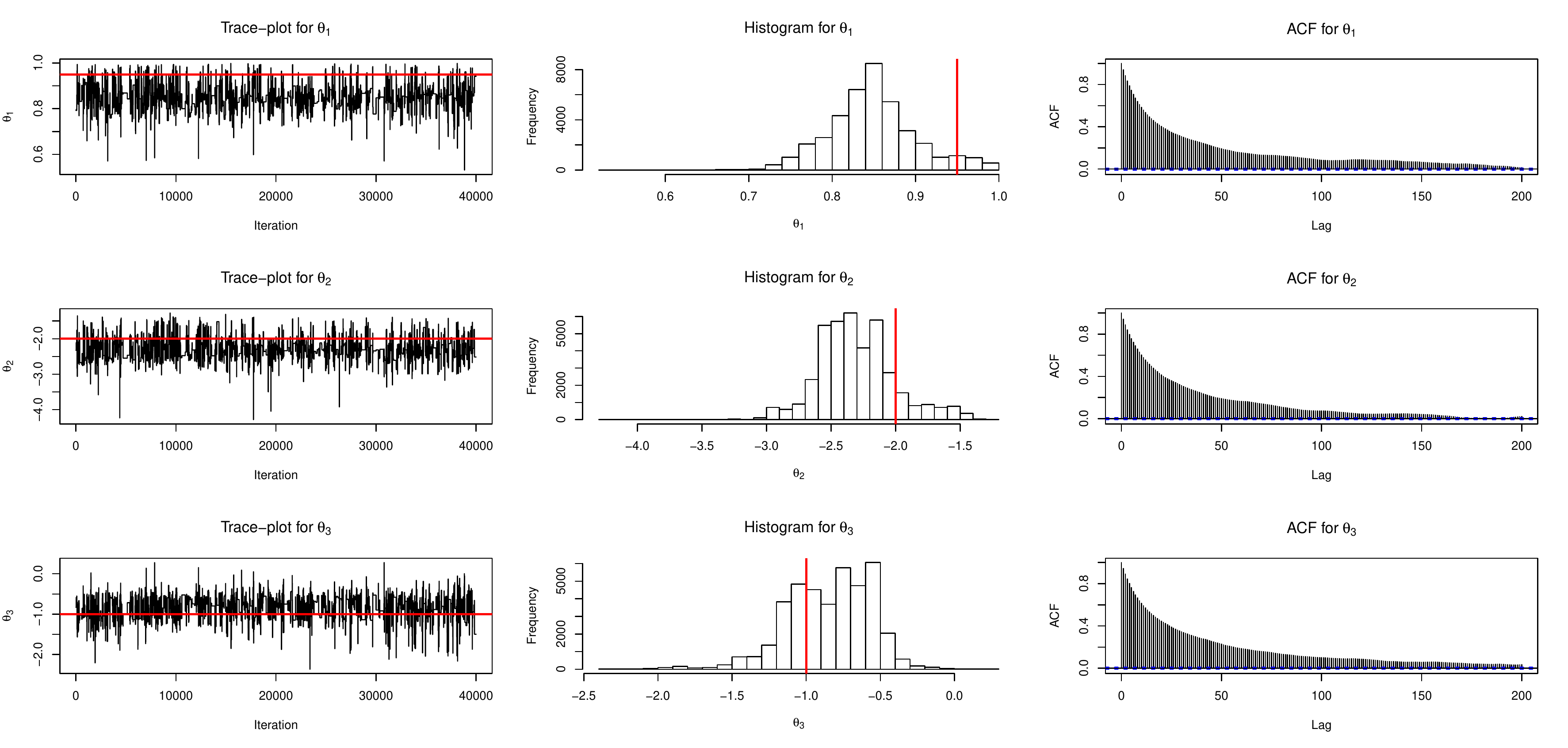}
\label{fig:sv-trace-wbslu}
\end{center}
\end{figure}
\begin{figure}[!ht]
\begin{center}
\caption {SV model: ABC-RW Sampler. Each row corresponds to parameters $\theta_1$ (top row), $\theta_2$ (middle row) and $\theta_3$ (bottom row)  and shows in order from left to right: Trace-plot, Histogram and Auto-correlation function. Red lines represent true parameter values. }
\includegraphics[scale=0.40]{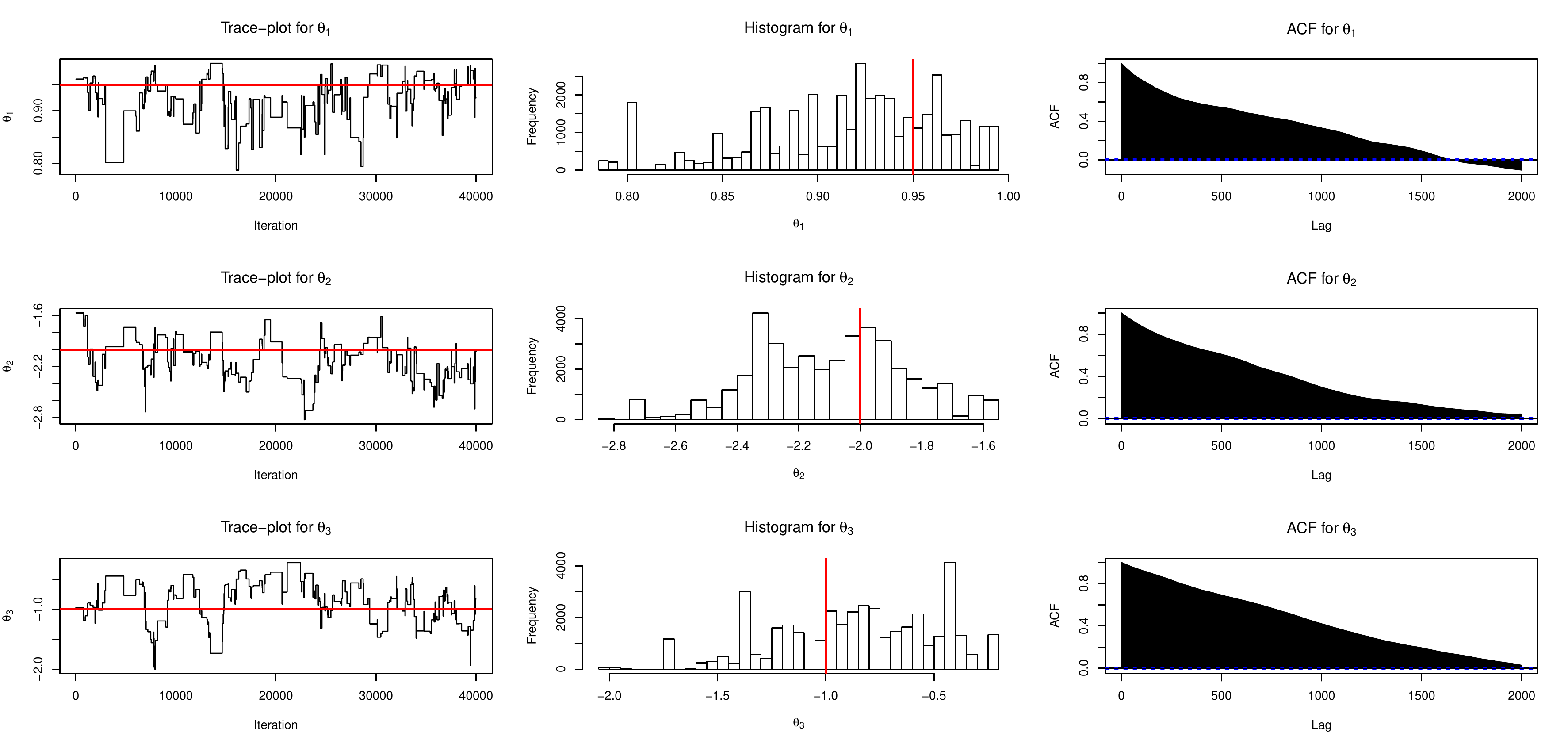}
\label{fig:sv-trace-abcrw}
\end{center}
\end{figure}
The major observation is that  AABC-U and ABSL-U are less sluggish than
ABC-RW, exhibiting  smaller auto-correlation values. 

In Figure~\ref{fig:sv-density-BSLIS} we compare the sample-based kernel smoothing posterior marginal density estimates for Exact, SMC, ABC-RW and AABC-U (top row) as well as Exact, BSL-IS and ABSL-U (bottom row). 
\begin{figure}[!ht]
\begin{center}
\caption {SV model: Estimated posterior marginal densities for each component. First row compares Exact, SMC, ABC-RW and AABC-U samplers. Second row compares Exact, BSL-IS and ABSL-U. Columns correspond to parameter's components, from left to right: $\theta_1$, $\theta_2$ and $\theta_3$.  }
\includegraphics[scale=0.40]{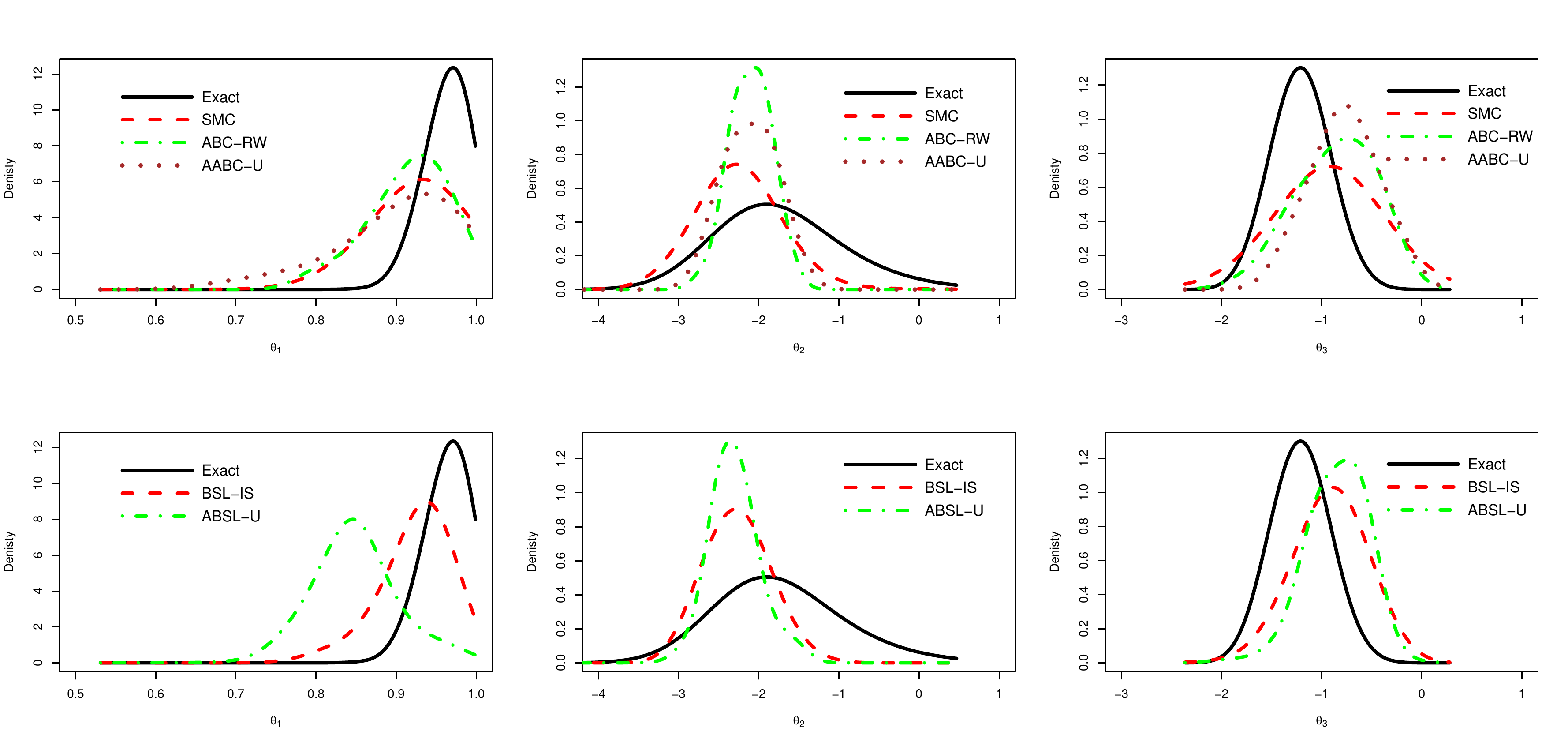}
\label{fig:sv-density-BSLIS}
\end{center}
\end{figure}
We note that all samples obtained from the approximate algorithms are  exact posterior (produced using PMCMC with 100 particles). Generally all ABC-based samplers perform similarly, on the other hand ABSL-U performs worse than generic BSL-IS in this run as it is shifted away from the exact posterior for $\th_1$ and $\th_3$.\\
To get more general conclusions we show average results in Table~\ref{table:sv} over 100 data replicates.
\begin{table} [!ht]
\begin{center}
\caption{ Simulation Results (SV model): Average Difference in mean, Difference in covariance, Total variation, square roots of Bias, variance and MSE, Effective sample size and Effective sample size per CPU time, for every sampling algorithm. }
\smallskip
\scalebox{1.0}{
\begin{tabular}{l | l l l || l l l || l l  }
\multicolumn{1}{c}{ }  &  \multicolumn{3}{c}{Diff with exact}   & \multicolumn{3}{c}{Diff with true parmater} & \multicolumn{2}{c}{Efficiency}\\
\hline
Sampler         & DIM & DIC & TV & $\sqrt{\mbox{Bias}^2}$ & $\sqrt{\mbox{VAR}}$ & $\sqrt{\mbox{MSE}}$ & ESS & ESS/CPU\\
\hline
SMC	&	0.232	&	0.0428	&	0.417	&	0.187	&	0.255	&	0.316	&	471	&	0.336\\
ABC-RW	&	0.210	&	0.0396	&	0.459	&	0.228	&	0.255	&	0.342	&	31	&	0.097	\\
ABC-IS	&	0.179	&	0.0439	&	0.460	&	0.196	&	0.219	&	0.294	&	30	&	0.090	\\
AABC-U	&	0.194	&	0.0447	&	0.424	&	0.212	&	0.217	&	0.304	&	1793	&	2.445	\\
AABC-L	&	0.189	&	0.0441	&	0.420	&	0.211	&	0.235	&	0.316	&	1659	&	2.253	\\
BSL-RW	&	0.200	&	0.0360	&	0.411	&	0.175	&	0.227	&	0.287	&	131	&	0.043	\\
BSL-IS	&	0.195	&	0.0362	&	0.404	&	0.175	&	0.225	&	0.285	&	346	&	0.113	\\
ABSL-U	&	0.229	&	0.0422	&	0.551	&	0.184	&	0.241	&	0.303	&	871	&	0.822	\\
ABSL-L	&	0.231	&	0.0410	&	0.548	&	0.197	&	0.240	&	0.311	&	843	&	0.817	\\
\hline
\end{tabular}
}
\label{table:sv}
\end{center}
\end{table}
Again we note that the proposed algorithms outperform the benchmark samplers by 8 times in ESS/CPU. Moreover AABC-U and AABC-L have very similar or smaller values for DIM, TV and MSE, which demonstrates that these samplers are much more efficient than standard methods and at the same produce as accurate (or more accurate) parameter estimates as generic algorithms.\\
ABSL-U and ABSL-L on the other hand did not perform well for this model, TV and MSE for these samplers are larger by 10\% than generic ones.
\subsection{Stochastic Volatility with $\al$-Stable errors}
\label{sec:sv_al}
As was pointed out in the previous sub-section, standard SV model assumes that the conditional distribution of the observed variables is Gaussian. Frequently, in financial time series, a large sudden drop occurs, thus raising serious doubts about  the latter assumption. Often, it is suggested to use heavy tailed distributions (instead of Gaussian) to model financial data. We consider a family of distributions named $\al$-Stable, denoted $Stab(\al,\beta)$, with two parameters $\al\in(0,2]$ (stability parameter) and $\beta\in[-1,1]$ (skew parameter). Two special cases are $\al=1$ and $\al=2$ which correspond to Cauchy and Gaussian distribution respectively, note that for $\al<2$ the distribution has undefined variance. We define the following SV model with $\al$-Stable errors with parameter
${\mathbf \th}=(\th_1,\th_2,\th_3,\th_4)^T \in\RR^4$:
\begin{equation}
\label{eq:sv_model}
\begin{split}
&x_{1} \sim \calN(0,1/(1-\th_1^2));\hsp v_i\overset{iid}{\sim}\calN(0,1);\hsp w_i\overset{iid}{\sim}Stab(\th_4,-1) ; \hsp i=\{1,\cdots,n\}, \\
&x_{i} = \th_1x_{i-1} + v_i; \hsp i=\{2,\cdots,n\}, \\
&y_{i} = \sqrt{\exp(\th_2 + \exp(\th_3)x_i)}w_i;  \hsp i=\{1,\cdots,n\}.
\end{split}
\end{equation}
This model is very similar to the simple SV with only difference that emission errors follow $\al$-Stable distribution with unknown stable parameter and fixed skew of $-1$. We generally prefer negative skew emission probability to model large negative financial returns. As in the previous simulation example $\th_2$ and $\th_3$ are unrestricted. The prior distribution for this model is (independently for each parameter):
\begin{equation}
\label{eq:sv_prior}
\begin{split}
& \th_1 \sim Unif(0,1), \\
& \th_2 \sim \calN(0,1), \\
& \th_3 \sim \calN(0,1), \\
& \th_4 \sim Unif(1.5,2).
\end{split}
\end{equation}  
We set the true parameters to $\th_1=0.95,\th_2=-2,\th_3=-1,\th_4=1.8$ and length of the time series $n=500$. The major challenge with this model is that there are no closed-form densities for $\al$-Stable distributions. Hence, most MCMC samplers, including PMCMC and ensemble MCMC, cannot be used  to sample from the posterior. However sampling from this family of distributions is feasible which makes it particularly amenable for simulation based methods like ABC and BSL.   
For summary statistics we use a 7-dimensional vector whose components are:
\begin{enumerate}
\item[(C1)] \#$\{i: y_i^2 > \mbox{quantile}(\by_0^2,0.99)\}$,
\item[(C2)] Average of $\by^2$,
\item[(C3)] Standard deviation of $\by^2$,
\item[(C4)] Sum of the first 5 auto-correlations of $\by^2$,
\item[(C5)] Sum of the first 5 auto-correlations of $\{\one_{\{y_i^2<\mbox{quantile}(\by^2,0.1)\}}\}_{i=1}^n$,
\item[(C6)] Sum of the first 5 auto-correlations of $\{\one_{\{y_i^2<\mbox{quantile}(\by^2,0.5)\}}\}_{i=1}^n$,
\item[(C7)] Sum of the first 5 auto-correlations of $\{\one_{\{y_i^2<\mbox{quantile}(\by^2,0.9)\}}\}_{i=1}^n$.
\end{enumerate}
Figures~\ref{fig:sv.stable-trace-wabcu},\ref{fig:sv.stable-trace-wbslu} and \ref{fig:sv.stable-trace-abcrw} show trace-plots, histograms and ACF function for AABC-U, ABSL-U and ABC-RW samplers respectively for each component (red lines correspond to the true parameters).
\begin{figure}[!ht]
\begin{center}
\caption {SV $\al$-Stable model: AABC-U Sampler. Each row corresponds to parameters $\theta_1$ (top row),  $\theta_2$ (second top row), $\theta_3$ (second bottom row), $\theta_4$ (bottom row)  and shows in order from left to right: Trace-plot, Histogram and Auto-correlation function. Red lines represent true parameter values. 
 }
\includegraphics[scale=0.40]{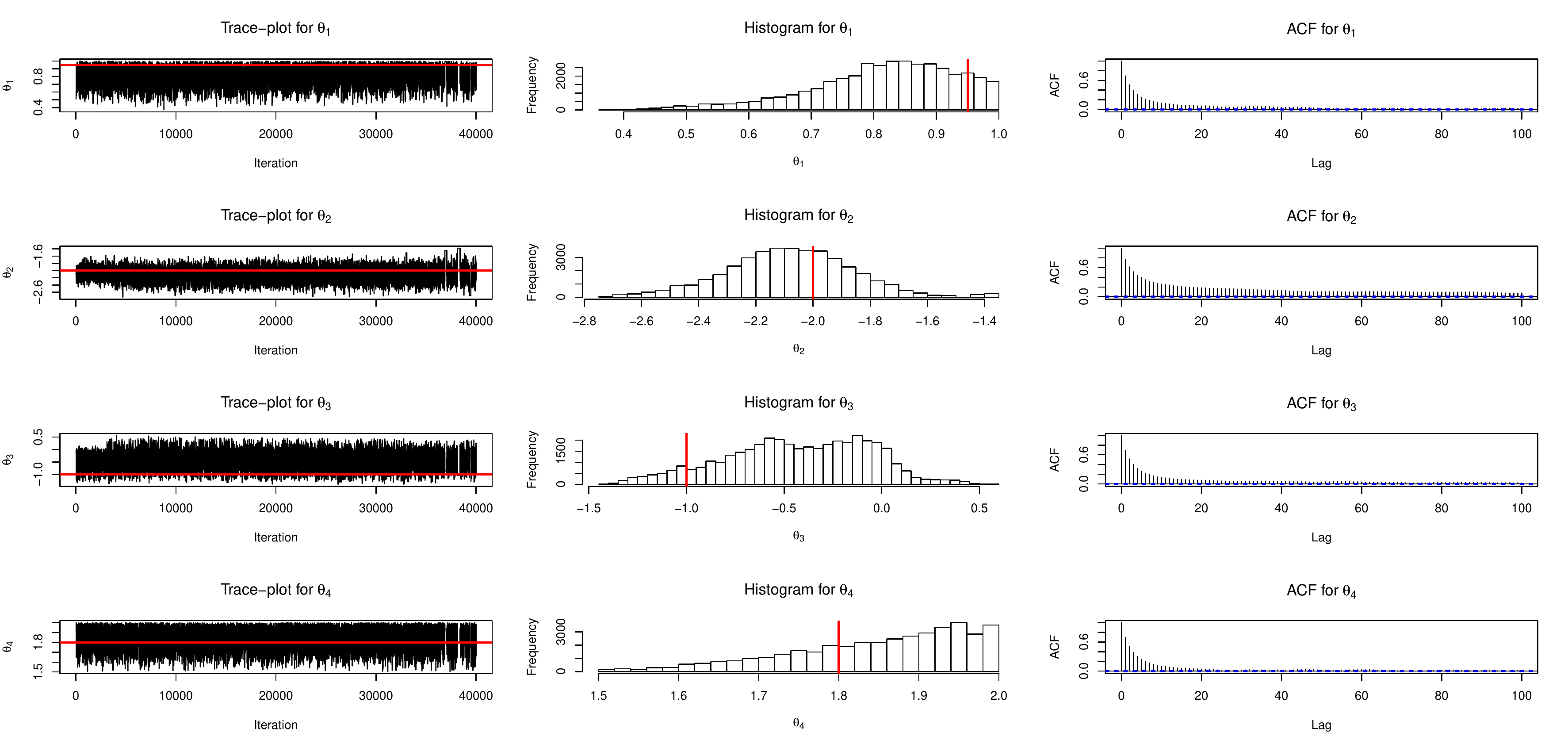}
\label{fig:sv.stable-trace-wabcu}
\end{center}
\end{figure}
\begin{figure}[!ht]
\begin{center}
\caption {SV $\al$-Stable model: ABSL-U Sampler. Each row corresponds to parameters $\theta_1$ (top row),  $\theta_2$ (second top row), $\theta_3$ (second bottom row), $\theta_4$ (bottom row)  and shows in order from left to right: Trace-plot, Histogram and Auto-correlation function. Red lines represent true parameter values.}
\includegraphics[scale=0.40]{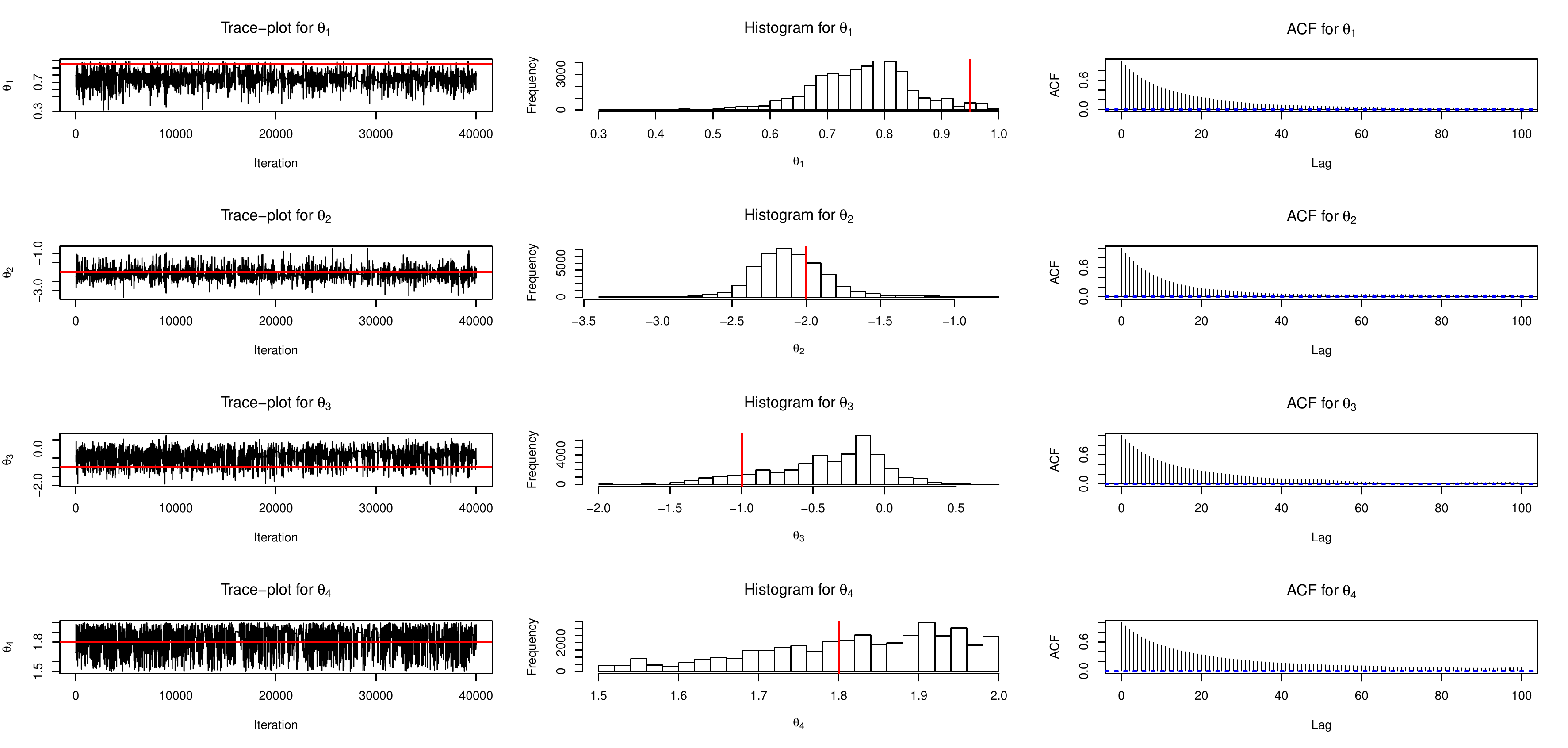}
\label{fig:sv.stable-trace-wbslu}
\end{center}
\end{figure}
\begin{figure}[!ht]
\begin{center}
\caption {SV $\al$-Stable model: ABC-RW Sampler. Each row corresponds to parameters $\theta_1$ (top row),  $\theta_2$ (second top row), $\theta_3$ (second bottom row), $\theta_4$ (bottom row)  and shows in order from left to right: Trace-plot, Histogram and Auto-correlation function. Red lines represent true parameter values.}
\includegraphics[scale=0.40]{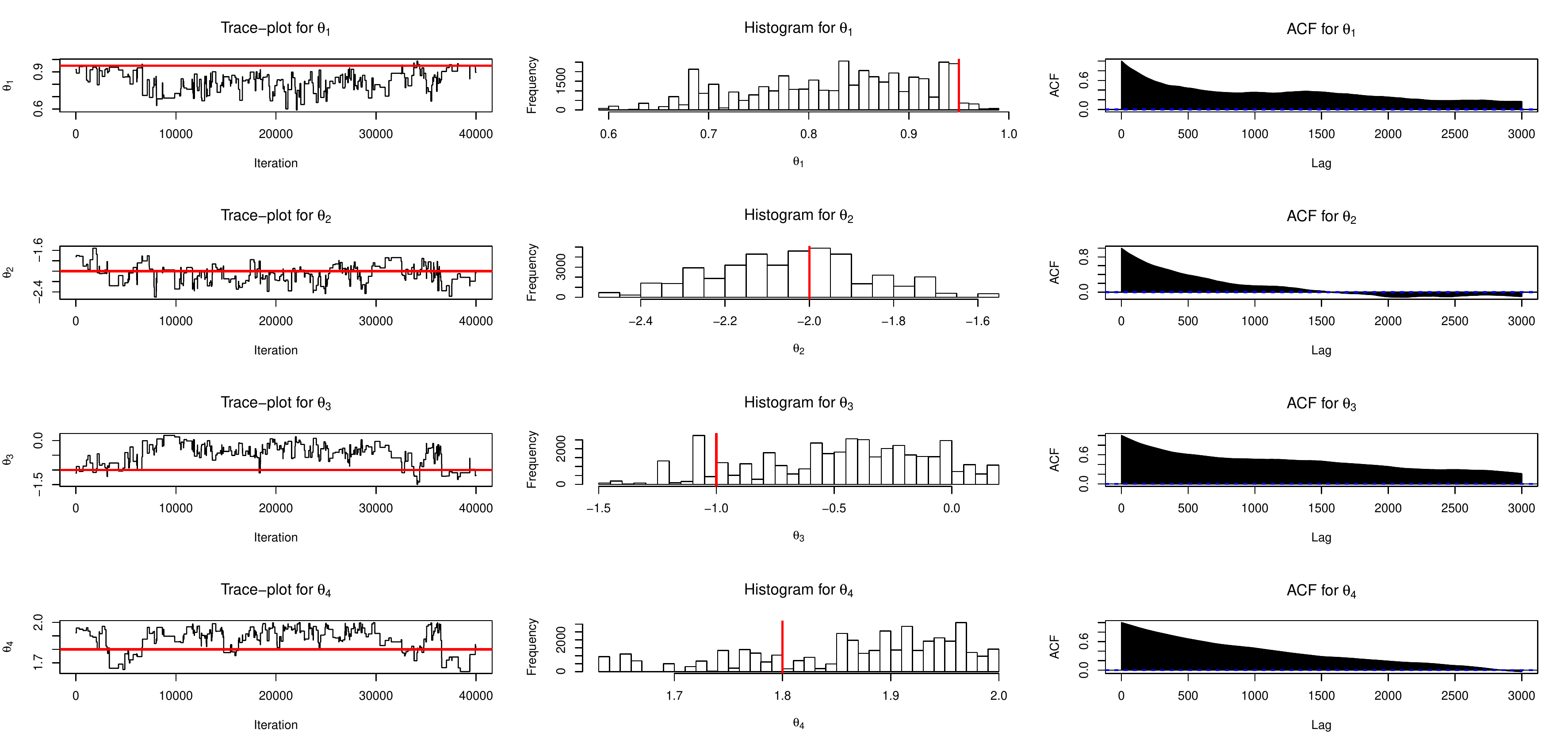}
\label{fig:sv.stable-trace-abcrw}
\end{center}
\end{figure}
As in previous examples the mixing of AABC-U and ABSL-U is much better than of ABC-RW. Since exact sampling is not feasible in this example we compare samplers to SMC (instead of exact samples), the plotted estimated densities are in Figure~\ref{fig:sv.stable-density-BSLIS}, here we have chosen BSL-IS over BSL-RW because it has better general performance in this model. 
\begin{figure}[!ht]
\begin{center}
\caption {SV $\al$-Stable model: Estimated densities for each component. First row compares SMC, ABC-RW and AABC-U samplers. Second row compares SMC, BSL-IS and ABSL-U. Columns correspond to parameter's components, from left to right: $\theta_1$, $\theta_2$, $\theta_3$ and $\theta_4$.}
\includegraphics[scale=0.40]{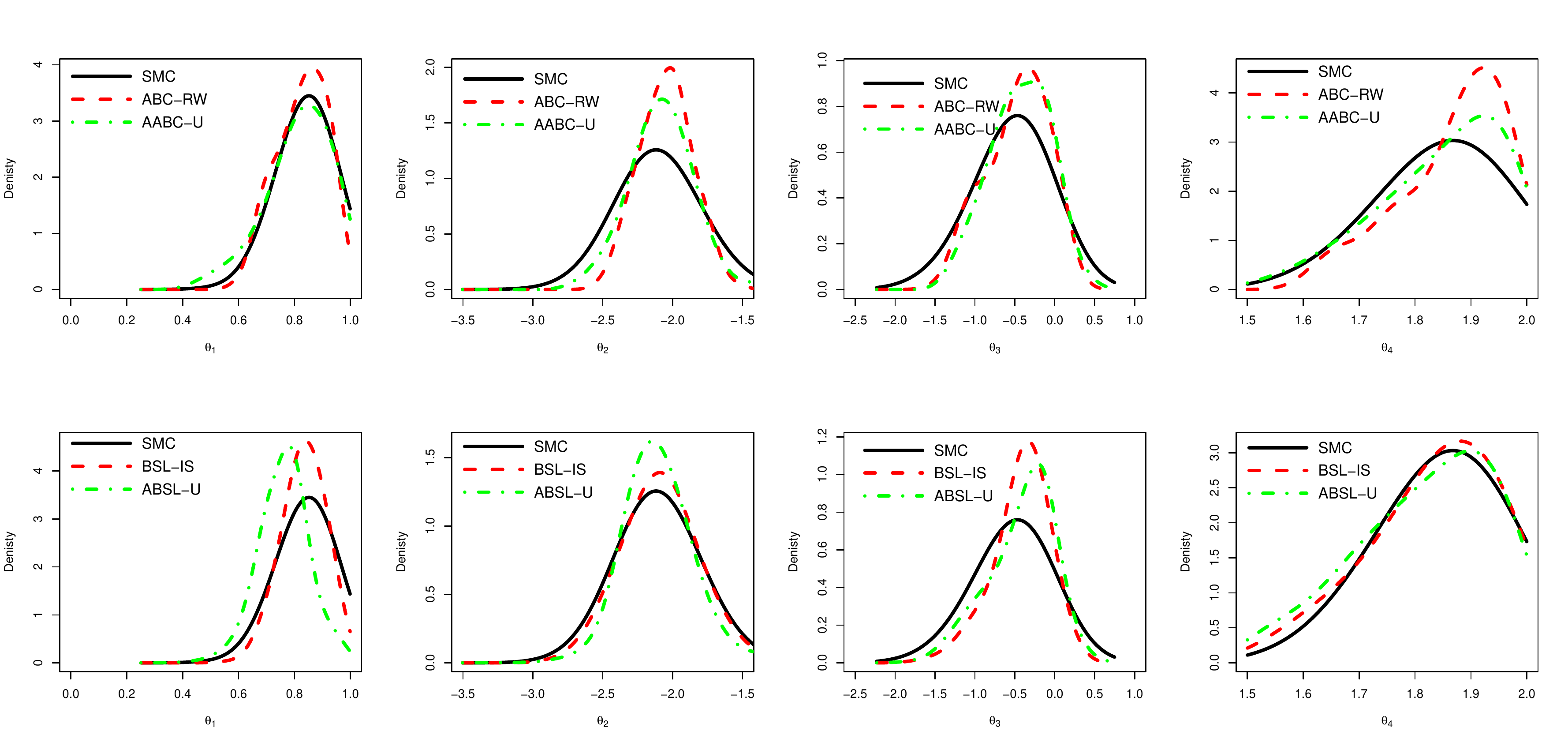}
\label{fig:sv.stable-density-BSLIS}
\end{center}
\end{figure}
Generally all simulation-based samplers have similar densities in this example.\\
For more general conclusions we show average results in Table~\ref{table:sv.stable} over 100 data replicates. Here to calculate DIM, DIC and TV, samplers are compared to SMC since exact draws cannot be obtained.
\begin{table} [!ht]
\begin{center}
\caption{ Simulation Results (SV $\al$-Stable model): Average Difference in mean, Difference in covariance, Total variation, square roots of Bias, variance and MSE, Effective sample size and Effective sample size per CPU time, for every sampling algorithm. In DIM, DIC and TV, samplers are compared to SMC. }
\smallskip
\scalebox{1.0}{
\begin{tabular}{l | l l l || l l l || l l  }
\multicolumn{1}{c}{ }  &  \multicolumn{3}{c}{Diff with SMC}   & \multicolumn{3}{c}{Diff with true parmater} & \multicolumn{2}{c}{Efficiency}\\
\hline
Sampler         & DIM & DIC & TV & $\sqrt{\mbox{Bias}^2}$ & $\sqrt{\mbox{VAR}}$ & $\sqrt{\mbox{MSE}}$ & ESS & ESS/CPU\\
\hline
SMC	&	0.000	&	0.0000	&	0.000	&	0.221	&	0.201	&	0.299	&	468	&	0.267\\
ABC-RW	&	0.078	&	0.0126	&	0.205	&	0.248	&	0.198	&	0.317	&	24	&	0.069\\
ABC-IS	&	0.082	&	0.0151	&	0.306	&	0.232	&	0.221	&	0.320	&	26	&	0.071\\
AABC-U	&	0.069	&	0.0124	&	0.170	&	0.250	&	0.183	&	0.310	&	1303	&	1.617	\\
AABC-L	&	0.069	&	0.0132	&	0.161	&	0.246	&	0.181	&	0.305	&	1256	&	1.546	\\
BSL-RW	&	0.044	&	0.0116	&	0.122	&	0.225	&	0.181	&	0.289	&	123	&	0.037	\\
BSL-IS	&	0.045	&	0.0103	&	0.125	&	0.226	&	0.177	&	0.287	&	285	&	0.084	\\
ABSL-U	&	0.063	&	0.0133	&	0.228	&	0.225	&	0.181	&	0.289	&	832	&	0.735	\\
ABSL-L	&	0.061	&	0.0140	&	0.230	&	0.236	&	0.183	&	0.299	&	757	&	0.671	\\
\hline
\end{tabular}
}
\label{table:sv.stable}
\end{center}
\end{table}
As in previous examples ESS/CPUs for AABC-U, AABC-L, ABSL-U and ABSL-L are roughly 8 times larger than benchmark algorithms. For this example looking at DIM, DIC and TV maybe misleading since approximated samplers are compared to another approximate sampler. Much more informative is MSE measure, it is very similar across ABC-based and BSL-based algorithms. Therefore we can conclude that proposed samplers perform very well in this example.   
\section{Data Analysis}
\label{sec:real}
For real world example we consider Dow-Jones index daily log returns from January 1, 2010 until December 31, 2018. The data were downloaded from Yahoo Finance\footnote{https://ca.finance.yahoo.com/} website. Given a time series of prices $P_i$, $i=1,\cdots, n$, log returns are calculated in the following way:
$$ r_i = \log(P_{i})-\log(P_{i-1}),\hsp i=2,\cdots, n.$$
The resulting time series is of length 2262. To make log returns more suitable for analysis, we standardize $r_t$ by subtracting its mean and then multiply each return by 200, so that absolute values were not too small, Figure~\ref{fig:real-ts} shows transformed returns. 
\begin{figure}[!ht]
\begin{center}
\caption {Dow Jones daily transformed log return for a period of Jan 2010 - Dec 2018.}
\includegraphics[scale=0.40]{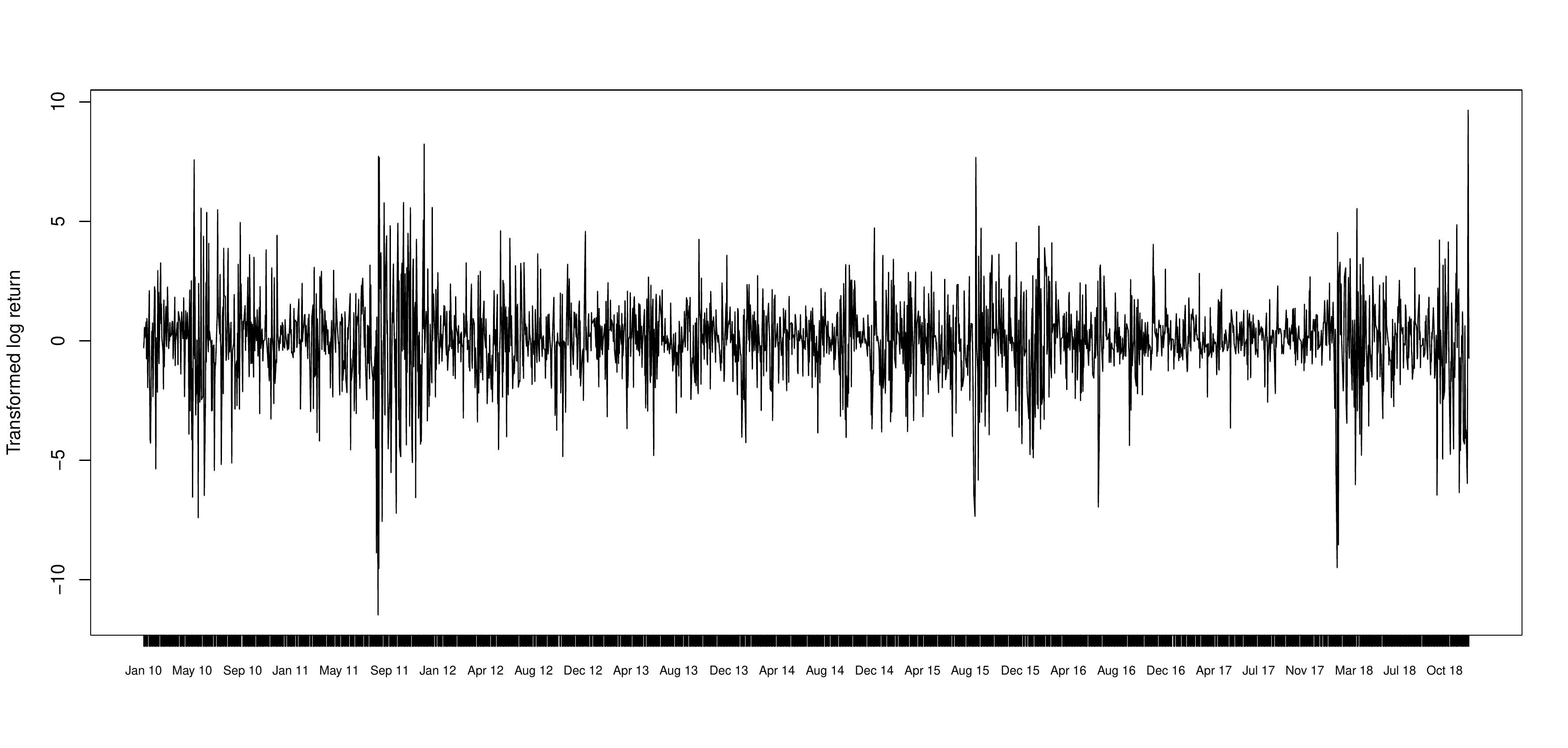}
\label{fig:real-ts}
\end{center}
\end{figure}
This time series ($\by_0$) has mean zero by construction, and its auto-correlations and partial auto-correlations are insignificant for any lag.  However, it is obvious that variances are correlated and there are alternating periods of low and high variability. This prompts us to use Stochastic Volatility model with $\al$-Stable errors as described in the previous section. Since the likelihood does not exist for this class of models, the simulation-based methods are probably the only available tools for the inference. The evolution of time series is described by equation \eqref{eq:sv_model}  and the parameter's prior is set as in equation \eqref{eq:sv_prior}. The skewed parameter of Stable distribution is fixed at value of $-1$. To estimate the posterior distribution we run AABC-U and ABLS-U samplers. The summary statistic for both methods is the same  7-dimensional vector defined in section \ref{sec:sv_al}. Each chain was run for 100 thousand iterations with last 80 thousands used for inference.
Figures~\ref{fig:real-abc} and \ref{fig:real-bsl} show trace-plots and histograms for AABC-U and ABSL-U samplers respectively for each parameter.
\begin{figure}[!ht]
\begin{center}
\caption {Dow Jones log returns: AABC-U Sampler. Every column corresponds to a particular parameter component from left to right: $\theta_1$, $\theta_2$, $\theta_3$, $\theta_4$ and shows trace-plot on top and histogram on bottom. }
\includegraphics[scale=0.40]{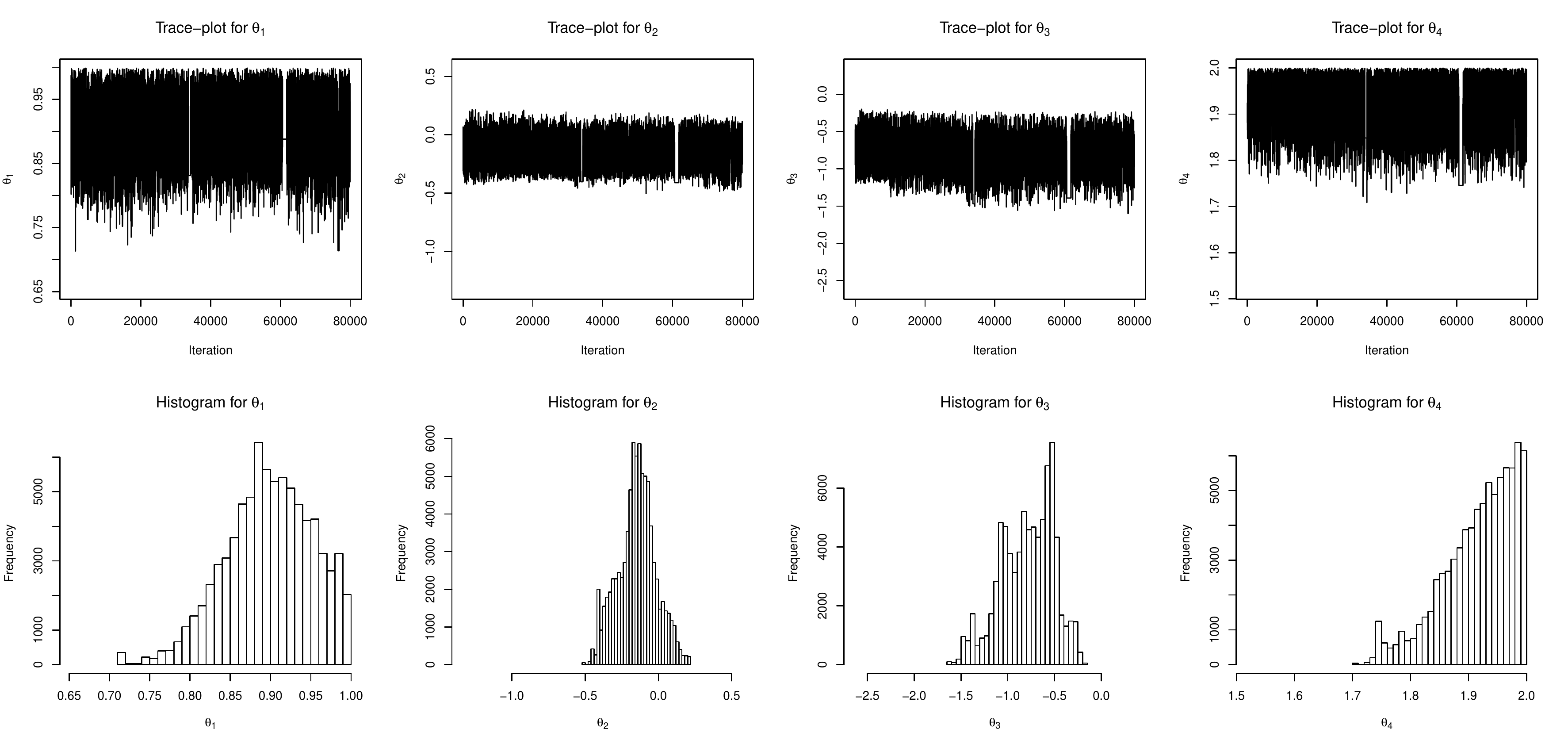}
\label{fig:real-abc}
\end{center}
\end{figure}
\begin{figure}[!ht]
\begin{center}
\caption {Dow Jones log returns: ABSL-U Sampler. Every column corresponds to a particular parameter component from left to right: $\theta_1$, $\theta_2$, $\theta_3$, $\theta_4$ and shows trace-plot on top and histogram on bottom. }
\includegraphics[scale=0.40]{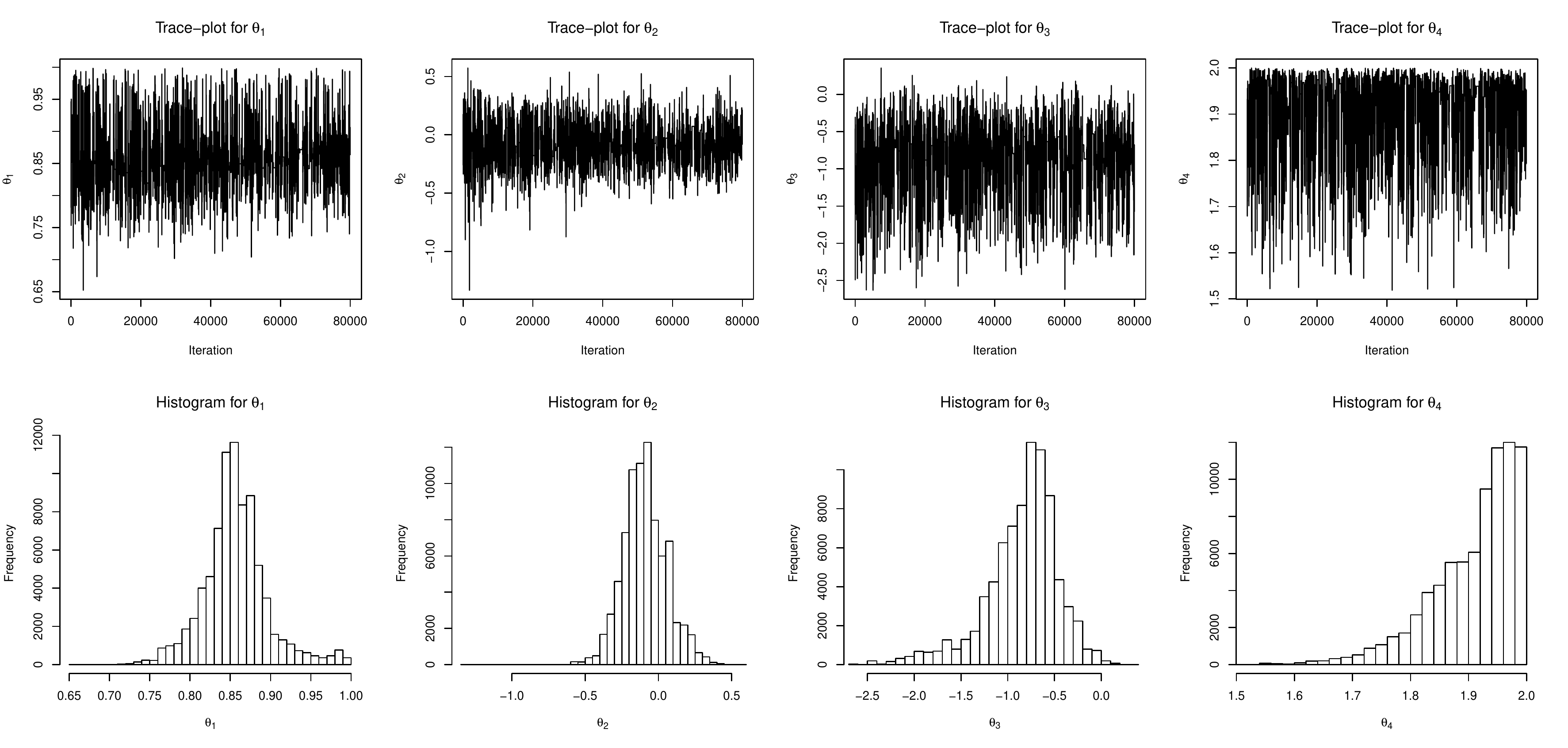}
\label{fig:real-bsl}
\end{center}
\end{figure}
The conclusions are in agreement with  the ones suggested by the simulation study. The mixing of AABC-U is generally better than of ABSL-U. However, posterior draws of ABSL-U for the first 3 components are uni-modal, symmetric and bell-shaped, which is not surprising since  the use of Gaussian priors within the BSL method yields  Gaussian posteriors due to conjugacy. Table~\ref{table:real} reports posterior mean and 95\% credible intervals for every parameter and for both samplers.
\begin{table} [!ht]
\begin{center}
\caption{ Dow Jones log return stochastic volatility: 95\% credible intervals and posterior averages for 4 parameters for two proposed samplers (AABC-U and ABSL-U). }
\smallskip
\scalebox{0.7}{
\begin{tabular}{l | r r r || r r r   }
\multicolumn{1}{c}{ }  &  \multicolumn{3}{c}{AABC-U}   & \multicolumn{3}{c}{ABSL-U} \\
\hline
Parameter   & 2.5\% Quantile & Average & 97.5\% Quantile & 2.5\% Quantile  & Average & 97.5\% Quantile  \\
\hline
$\th_1$	&	0.787	&	0.899	&	0.990	&	0.775	&	0.856	&	0.959	\\
$\th_2$	&	-0.411	&	-0.147	&	0.112	&	-0.369	&	-0.092	&	0.222	\\
$\th_3$	&	-1.405	&	-0.790	&	-0.304	&	-1.858	&	-0.841	&	-0.206	\\
$\th_4$	&	1.758	&	1.916	&	1.997	&	1.721	&	1.909	&	1.996	\\
\hline
\end{tabular}
}
\label{table:real}
\end{center}
\end{table}
AABC-U and ABSL-U produce similar results. We see that the estimated correlation between adjacent variables in the hidden layer  is about $0.9$ and the estimate of $\al$-Stable emission noise is $1.91$.  This model can produce more extreme values than those predicted by one with standard Gaussian noise.  

\section{Theoretical Justifications}
\label{sec:theory}
In this section we show that the novel approximated ABC MCMC and BSL samplers with independent proposals exhibit ergodic properties in a long run. In other words, we want to show that as number of MCMC iterations increases marginal distribution of $\{\th^{(t)}\}$ converges to appropriate posterior distribution in total variation and sample averages converge to the true expectations. \\
We start by reviewing our notation. Let $p(\th), q(\th)$ represent the prior and proposal distributions for $\th \in \Th$ respectively. For AABC we define a function $h(\th)$ as $P(\delta<\eps|\th)$ where $\delta = \delta(\by,\by_0)$ and $\by \sim f(\by|\th)$. Then given a proposed $\z^*$ the acceptance probability is:
\begin{equation} 
\begin{split}
a(\th,\z^*)&=\min(1,\al(\th,\z^*)), \\
\al(\th,\z^*) &= \frac{p(\z^*)q(\th)h(\z^*)}{p(\th)q(\z^*)h(\th)}.
\end{split}
\end{equation}
This MH procedure defines an exact transition kernel which we call $P(\cdot,\cdot)$. Since $h(\th)$ is not available in closed form we will estimate it using k-nearest-neighbor approach. \\
Let $\zn = \{\tilde \z_n,\one_{\{\tilde \delta_n<\eps\}}\}_{n=1}^N$ represent $N$ independent samples from $q(\z)P(\one_{\{\delta<\eps\}}|\z)$ for AABC. Actually $\zn$ contains past generated samples that were saved before $N$th iteration. Given $\th$ and $\z^*$ we apply kNN to approximate $h(\th)$ and $h(\z^*)$ by calculating local weighted averages of $\one_{\{\tilde \delta_n<\eps\}}$ for $\tilde \z_n$ that are close to $\th$ or $\z^*$.  We denote such estimate $\hat h(\th;\zn)$, and the probability of proposal acceptance for this perturbed algorithm (more on perturbed MCMC can be found in \cite{roberts1998convergence, pillai2014ergodicity, johndrow2017error}) is:
\begin{equation} 
\begin{split}
\hat a(\th,\z^*;\zn)&=\min(1,\hat \al(\th,\z^*;\zn)), \\
\hat \al(\th,\z^*;\zn) &= \frac{p(\z^*)q(\th)\hat h(\z^*;\zn)}{p(\th)q(\z^*)\hat h(\th;\zn)}.
\end{split}
\end{equation}
The approximate kernel transition is $\hat P_N(\cdot,\cdot)=E_{\zn}\left[ \hat P_N(\cdot,\cdot;\zn)\right]$, the goal is to show that as $N \to \infty$ the distance between this transition and the exact one converges to zero, where distance is defined as:

\begin{equation}
\|\hat P_N(\cdot,\cdot)-P(\cdot,\cdot)\|=\sup_{\th}\|\hat P_N(\th,\cdot)-P(\th,\cdot)\|_{TV},
\end{equation}
where the last distance is "total variation" distance between two measures. First we show that under strong consistency assumption of $\hat h(\th;\zn)$, perturbed kernel converges to the exact one.
\begin{theorem}
\label{perturb_mcmc}
Suppose $\Th$ is compact, $\sup_{\th}\|\hat h(\th;\zn)-h(\th)\|\to 0$ with probability 1 and $h(\th)>0$ for all $\th \in \Th$. Then for any $\eps>0$ there exists $C$ such that for all $N>C$, $\|\hat P_N- P\|<\eps$. 
\end{theorem}  
Next let $\calP_\eps=\{\hat P_N:\|\hat P_N-P\|<\eps\}$  be a collection of perturbed kernels each $\eps$ distance from the exact kernel. For illustration consider an example when auxiliary set $\zn$ grows with number of iterations, in this case at each iteration a new kernel $\hat P_N \in \calP_\eps$ is used in the chain. We want to show that this procedure will results in ergodic chain with appropriate convergence results. For most of the presented results below we refer to the work of \cite{johndrow2015optimal} on convergence properties of perturbed kernels.\\
To obtain useful convergence results we need to make additional Doeblin Condition assumption about the exact kernel $P$:
\begin{definition}[Doeblin Condition]
Given a kernel $P$, there exists $0<\al<1$ such that
$$ \sup_{(\th,\z^*)\in \Th \times \Th}\|P(\th,\cdot)-P(\z^*,\cdot)\|_{TV}<1-\al. $$
\end{definition}
We also choose $\eps$ so that $\al^*= \al + 2\eps < 1$ and $\eps<\al/2$ which by Remark 2.1 in \cite{johndrow2015optimal} guarantees that every member of $\calP_\eps$ satisfies Doeblin Condition with $\al = \al^*$ and has a unique invariant measure. Thus we define the following 3 assumptions:
\begin{enumerate}
\item[({\bf A1})] Exact transition kernel $P$ satisfies satisfies the Doeblin Condition,
\item[({\bf A2})]  For any  $\hat P \in \calP_\eps$, $\|\hat P-P\|<\eps$,
\item[({\bf A3})]  $\eps < \min(\al/2,(1-\al)/2)$.
\end{enumerate}
Now, let $\mu$ be invariant measure of the exact kernel $P$, and the perturbed chain $\th^{(0)},\th^{(1)},\cdots,\th^{(t)}$ is a Markov chain with $\th^{(0)}\sim \nu=\mu_0$. Also define marginal distribution of $\th^{(t)}$ denoted by $\mu_t$, $t=1,2,,$ and equal to $\mu_t = \nu \hat P_0 \hat P_1\cdots \hat P_{t}$ with each $\hat P_t \in \calP$, $t=1,2,\cdots$ and $\hat P_0$ being identity transition (for convenience). First we need to examine the total variation distance between $\mu$ and average measure $\sum_{t=0}^{M-1}\mu_t /M$, in other words:
\begin{equation}
\left\|\mu - \frac{\sum_{t=0}^{M-1}\nu \hat P_{0} \cdots \hat P_{t}  }{M}\right\|_{TV},\hsp \mbox{where} \hsp \hat P_{0}=I.
\end{equation} 
Then we have the following important convergence result:
\begin{theorem}
\label{conv_result_1}
Suppose that  ({\bf A1}), ({\bf A2}) and ({\bf A3}) are satisfied. Let $\nu$ be any probability measure on $(\Th,\calF_0)$, then 
\begin{equation}
\left\| \mu - \frac{\sum_{t=0}^{M-1}\nu \hat P_{0} \cdots \hat P_{t}  }{M}\right\|_{TV} \leq \frac{(1-(1-\al)^M)\|\mu-\nu\|_{TV}}{M\al} - \frac{\eps(1-(1-\al)^M)}{M\al^2}+\frac{\eps}{\al},
\end{equation}
which implies that this difference can be arbitrary small for sufficiently large $M$ and small enough $\eps$.
\end{theorem}
Next we focus on the following mean squared error (MSE):
$$ E\left[ \left(\mu f - \frac{\sum_{t=0}^{M-1}f(\th^{(t)})}{M}\right)^2 \right], $$
where $f$ is bounded function and $\mu f = E_{\mu}[f(\th)]$. The main objective here is to find the upper bound for this MSE when perturbed MCMC is used and how it depends on the sample size $M$. To obtain the main result we introduce the following lemma:
\begin{lemma}
\label{lemma_cov}
Suppose:  ({\bf A2}) and ({\bf A3}) are satisfied;  $\th^{(0)} \sim \nu$, where $\nu$ is a probability distribution; $\mu_t=\nu \hat P_1 \cdots \hat P_t$ is the marginal distribution of $\th^{(t)}$, $t=1,2,\cdots$. Let $f(\th)$ and $g(\th)$ be bounded functions with $|f|=\sup_{\th}f(\th)$ and $|g|=\sup_{\th}g(\th)$. Then
$$cov(f(\th^{(k)}),f(\th^{(j)})) \leq 8|f||g|(1-\al^*)^{|k-j|}.$$ 
\end{lemma}
The next important convergence results follows (similar to Theorem 2.5 of \cite{johndrow2015optimal}):
\begin{theorem}[Approximation of MSE]
\label{conv_result_2}
Suppose that  ({\bf A1}), ({\bf A2}) and ({\bf A3}) are satisfied.  Let $\mu$ denote the invariant measure of $P$, $f(\th)$ be a bounded function and $\th^{(0)} \sim \nu$, where $\nu$ is a probability distribution . Then
\begin{equation}
\begin{split}
& E\left[ \left(\mu f - \frac{1}{M}\sum_{t=0}^{M-1} f(\th^{(t)})   \right)^2 \right] \\
& \leq 4|f|^2 \left( \frac{(1-(1-\al)^M)}{M\al} - \frac{\eps(1-(1-\al)^M)}{M\al^2} + \frac{\eps}{\al} \right)^2 \\ 
& +  8|f|^2 \left( \frac{1}{M} +  \frac{2}{(\al^*)^2} \left( \frac{(1-\al^*)^{M+1}-(1-\al^*)}{M^2} + \frac{(1-\al^*)-(1-\al^*)^2}{M} \right) \right).
\end{split}
\end{equation}
In other words this expectation can be made arbitrary small for sufficiently large $M$ and small enough $\eps$.
\end{theorem}

Based on these theorems we   obtain convergence results for AABC and ABSL algorithms. To that end, we consider the following  assumptions:
\begin{itemize}
\item[{\bf (B1)}] $\Th$ is a compact set.
\item[{\bf(B2)}] $q(\th)>0$ continuous density of independent proposal distribution.
\item[{\bf(B3)}] $p(\th)>0$ continuous density of prior distribution.
\item[{\bf(B4)}] $h(\th)$ continuous function of $\th$. 
\item[{\bf(B5)}] In kNN estimation assume that $K(N)=\sqrt{N}$ with uniform or linear weights.
\item[{\bf(B6)}] $E[s^j|\th]$ and $E[s^js^k|\th]$ are continuous functions of $\th$ for every $1\leq j,k\leq p$ with $s^j$ representing $j$th component of summary statistic $s$.
\item[{\bf(B7)}] $Var[s^j|\th]$ and $Var[s^j s^k|\th]$ are bounded functions.
\item[{\bf(B8)}] $|\Sigma_{\th}|>a_0$ where $\Sigma_{\th}= Var(s|\th)$ for every $\th\in \Th$.
\end{itemize}

\begin{theorem}[Ergodicity of AABC]
\label{abc_conv}
Consider the proposed AABC sampler with $\eps$ threshold and let: $p(\th)$ denote the prior measure on $\Th$, $\zn$ denote simulated pairs $\{\tilde \z_n,\one_{\{\tilde \delta_n<\eps\}}\}_{n=1}^N$ with $\tilde \z_n \sim q(\z)$ $\forall n$. Assume  {\bf (B1)}-{\bf(B5)} hold.
Then for sufficiently large $N$ (number of past simulations) and $M$ (number of chain iterations), assumptions {\bf (A1)}-{\bf (A3)} are satisfied and  the results established in Theorems~\ref{conv_result_1} and \ref{conv_result_2} follow.
\end{theorem} 

\begin{corollary}[Ergodicity of ABSL]
\label{sbl_conv}
Assume that {\bf (B1)}-{\bf(B8)} hold. Let $p(\th)$ be the prior distribution on $\Th$, $h(\th) = \calN\left(s_0;\mu_{\th},\Sigma_{\th}\right)$, $\zn$ the set of simulated pairs $\{\tilde \z_n,\{\tilde s_n^{(j)}\}_{j=1}^m\}_{n=1}^N$.
Then for sufficiently large $N$ (number of past simulations) and $M$ (number of chain iterations), assumptions {\bf (A1)}-{\bf (A3)} are satisfied and  the results established in Theorems~\ref{conv_result_1} and \ref{conv_result_2} follow.
\end{corollary} 
To prove the results above we will utilize the following two theorems, one is about the strong uniform consistency of kNN estimators the later one is about uniform ergodicity of Hastings algorithm with independent proposal.
\begin{theorem}[Uniform Consistency of kNN - \cite{cheng1984strong}]
\label{knn}
Given independent $\{\tilde \z_n,\tilde \delta_n\}_{n=1}^{N}$, let $\Th$ be support of distribution of $\tilde \z$, $h(\tilde \z)=E(\tilde \delta|\tilde \z)$ and $\hat h_N(\tilde \z)=\sum_{j=1}^N W_{Nj}\tilde \delta_j$ (kNN estimator) (here $j$ are permuted indices that order distances between $\tilde \z_n$ and $\tilde \z$ from smallest to largest). Suppose weights $W_{Nj}$ satisfy 
\begin{itemize}
\item[(i)] $\sum_{j=1}^N W_{Nj} =1 $,
\item[(ii)] $W_{Nj}=0$ for $j>K$, and $K=K(N)$ with $K \to \infty$ and $K/N \to 0$,
\item[(iii)] $\sup_N K \max_j W_{Nj} < \infty$.
\end{itemize} 
If
\begin{itemize}
\item[(i)] $\Th$ is compact,
\item[(ii)] $h(\tilde \z)$ is continuous function,
\item[(iii)] $Var(\tilde \delta|\tilde \z)$ is bounded random variable,
\item[(iv)] $K(N)$ satisfies $K/\sqrt{N}\log(N) \to \infty$,
\end{itemize}
then $\sup_{\tilde \z \in \Th}|\hat h_N(\tilde \z)-h(\tilde \z)|\to 0$ with probability 1.
\end{theorem}
Note that the uniform and linear weights satisfy $W_{Nj}$ assumptions above.
\begin{theorem}[Independent Metropolis sampler  - \cite{mengersen1996rates}]
\label{hast-unif}
Suppose $\th^{(t)}$ is a MH Markov Chain with invariant distribution $\pi(\th)$, independent proposal $q(\th)$ and acceptance probabilities $a(\th,\z^*) = \min\left( 1,\frac{\pi(\z^*)q(\th)}{\pi(\th)q(\z^*)}\right)$.\\
If there exists $\beta>0$ such that $q(\th)/\pi(\th)>\beta$ for all $\th\in \Th$, then the algorithm is \textit{uniformly ergodic} so that $\|P^n(\th,\cdot)-\pi\|_{TV}<(1-\beta)^n$ (here $P^n(\th,\cdot)$ is conditional distribution of $\th^{(n)}$ given $\th^{(0)}=\th$).
\end{theorem}
\subsection{Proofs of theorems}
\begin{proof}[Proof of Theorem~\ref{perturb_mcmc}]
Note that $\sup_{\th}\|\hat h(\th;\zn)-h(\th)\|\to 0$ w.p.1 implies that for all $\th$ and $\z^*$ in $\Th$: 
$$\hat h(\th;\zn) \overset{p}{\to} h(\th),$$
$$\hat h(\z^*;\zn) \overset{p}{\to} h(\z^*),$$
therefore by Slutsky's theorem we obtain
$$\frac{\hat h(\z^*;\zn)}{\hat h(\th;\zn)} \overset{p}{\to} \frac{h(\z^*)}{h(\th)},$$ for all $(\th,\z^*)$ in $\Th \times \Th$.
therefore 
$$\hat \al(\th,\z^*;\zn)=\frac{p(\z^*)q(\th)\hat h(\z^*;\zn)}{p(\th)q(\z^*)\hat h(\th;\zn)} \overset{p}{\to} \frac{p(\z^*)q(\th)h(\z^*)}{p(\th)q(\z^*)h(\th)}=\al(\th,\z^*).$$
Since $\min(1,x)$ is a continuous function, Continuous Mapping Theorem implies that
$$\hat a(\th,\z^*;\zn)=\min(1,\hat \al(\th,\z^*;\zn)) \overset{p}{\to} \min(1,\al(\th,\z^*))=a(\th,\z^*).$$
Note that this not just a point-wise convergence, but uniform convergence in probability so that one $C$ will work for all $(\th,\z^*)$. That is, for any $(\th,\z^*)$, $\delta>0$ and $\eps>0$ there exists $C$ such that for all $N>C$, $P(|\hat a(\th,\z^*;\zn)-a(\th,\z^*)|>\delta)<\eps$.\\
Another important observation is that (fixing $\th$, $\z^*$ and letting $a(\th,\z^*)=a$ and $\hat a(\th,\z^*;\zn)=\hat a$ for convenience)
\begin{equation}
\label{eq:exp_bound}
\begin{split}
E_{\zn}(|\hat a-a|)&=\int |\hat a-a| dF(\zn)=\int_{|\hat a-a|<\delta}|\hat a-a| dF(\zn)+\int_{|\hat a-a|\geq\delta}|\hat a-a| dF(\zn)\leq \\
& \leq \delta + \int_{|\hat a-a|\geq\delta} dF(\zn)\leq \delta + \eps.
\end{split}
\end{equation}
Because $|\hat a-a|\leq 1$ and applying definition of convergence in probability. The above inequality shows that we can make this expected value arbitrary small by taking large enough $N$, moreover this result is uniform so one $N$ will work for all $\th$ and $\z^*$.\\
Next we focus on the distance between two transition kernels, this discussion is similar to the proof of Corollary 2.3 in \cite{alquier2016noisy}. Observe that (using independent proposals):
\begin{equation*}
\begin{split}
& P(\th,d\z^*)=q(\z^*)a(\th,\z^*)d\z^* + \delta_{\th}(\z^*)r(\th), \\
& \hat P_N(\th,d\z^*)=\int q(\z^*) \hat a(\th,\z^*;\zn)d\z^* dF(\zn) + \delta_{\th}(\z^*)\hat r_N(\th),
\end{split}
\end{equation*}
where $r(\th)=1-\int q(\z^*)a(\th,\z^*)d\z^*$ and $\hat r_N(\th)=1-\int \int q(\z^*)a(\th,\z^*)d\z^* dF(\zn)$.
Fix $\th \in \Theta$, and noting that total variation between two probability distributions that have densities is also equal to:
$$\|\pi-\hat \pi\|_{TV} = 0.5\int |\pi(\th)-\hat \pi(\th)| d\th.$$
Therefore
\begin{equation}
\begin{split}
P(\th,d\z^*)-\hat P_N(\th,d\z^*) = & \int q(\z^*)(a(\th,\z^*)-\hat a(\th,\z^*;\zn))dF(\zn) \\
& + \delta_{\th}(d\z^*)\int \int q(t)(a(\th,t)-\hat a(\th,t;\zn))dF(\zn)dt,
\end{split}
\end{equation}
and it follows that 
\begin{equation}
\begin{split}
\|P(\th,d\z^*)-\hat P_N(\th,d\z^*)\|_{TV} \leq & 0.5 \left\{ \left| \int \int q(\z^*)(a(\th,\z^*)-\hat a(\th,t;\zn))dF(\zn)d\z^* \right| \right. \\
& \left. + \left| \int \int q(t)(a(\th,t)-\hat a(\th,t;\zn))dF(\zn)dt \right| \right\} \\
& = \left| \int \int q(t)(a(\th,t)-\hat a(\th,t;\zn))dF(\zn)dt \right| \\
& \leq \int \int q(t)\left| a(\th,t)-\hat a(\th,t;\zn) \right| dF(\zn)dt \leq \delta + \eps\\
\end{split}
\end{equation}
for any $\eps>0$ and $\delta>0$ and large enough $N$ by \eqref{eq:exp_bound}. Since this result is true for any $\th \in \Th$ we finally get the main result:
\begin{equation}
\sup_{\theta}\|\hat P_N(\th,d\z^*)-P(\th,d\z^*)\|_{TV}\leq \delta+\eps
\end{equation}
\end{proof}

\begin{proof}[Proof of Theorem~\ref{conv_result_1}]
We generally follow the proof of Theorem 2.4 in \cite{johndrow2015optimal}. First observe that:
$$ \nu \hat P_{0} \cdots \hat P_{M} -\mu P^M = (\nu-\mu)P^M + \sum_{t=0}^{M-1}\nu \hat P_0\cdots \hat P_{t}(\hat P_{t+1} -P)P^{M-t-1}. $$
By Assumptions 2 and 3, we get:
$$ \|\nu \hat P_0 \cdots \hat P_{t} \hat P_{t+1} - \nu \hat P_0 \cdots \hat P_{t} P\|_{TV} \leq \eps,$$
and 
$$ \|\nu \hat P_0 \cdots \hat P_{t} \hat P_{t+1} P^{n-t-1} - \nu \hat P_0 \cdots \hat P_{t} P P^{M-t-1}\|_{TV} \leq \eps(1-\al)^{M-t-1}.$$
Using these results, the triangular inequality and formula for sum of finite geometric series we establish that:
\begin{equation}
\begin{split}
\|\nu \hat P_0 \cdots \hat P_{M} - \mu P^M\|_{TV} \leq & \|\mu P^M-\nu P^M\|_{TV} + \sum_{t=0}^{M-1} \|\nu \hat P_0 \cdots \hat P_{t} \hat P_{t+1} P^{M-t-1} - \nu \hat P_0 \cdots \hat P_{t} P P^{M-t-1}\|_{TV} \\
\leq & (1-\al)^M \|\mu-\nu\|_{TV} + \eps \sum_{t=0}^{M-1}(1-\al)^{M-t-1} \\
= &  (1-\al)^M \|\mu-\nu\|_{TV} + \eps \frac{1-(1-\al)^M}{\al}.
\end{split}
\end{equation}
Finally we get the main result using that fact that $\mu$ is invariant for $P$ (again using sum of finite geometric series)
\begin{equation}
\begin{split}
\left\|\mu - \frac{\sum_{t=0}^{M-1}\nu \hat P_{0} \cdots \hat P_{t}  }{M}\right\|_{TV} = &
\left\|\frac{\sum_{t=0}^{M-1} \mu P^t}{M} - \frac{\sum_{t=0}^{M-1}\nu \hat P_{0} \cdots \hat P_{t}  }{M}\right\|_{TV} \\ \leq & \frac{1}{M}\sum_{t=0}^{M-1} \|\mu P^t - \nu \hat P_{0} \cdots \hat P_{t} \|_{TV} \\
\leq & \frac{1}{M}\sum_{t=0}^{M-1} \left( (1-\al)^t \|\mu-\nu\|_{TV} + \eps \frac{1-(1-\al)^t}{\al} \right) \\ = & \frac{(1-(1-\al)^M)\|\mu-\nu\|_{TV}}{M\al} - \frac{\eps(1-(1-\al)^M)}{M\al^2}+\frac{\eps}{\al}.
\end{split}
\end{equation}
\end{proof}

\begin{proof}[Proof of Lemma~\ref{lemma_cov}]
Without loss of generality we assume that $k>j$, next define:
$$ \tilde f(\th^{(j)}) = f(\th^{(j)}) - \mu_j f, $$
$$ \tilde g(\th^{(k)}) = g(\th^{(k)}) - \mu_k g, $$
so that $E[\tilde f(\th^{(j)})]=E[\tilde g(\th^{(k)})]=0$. Then we get the following
\begin{equation}
\label{eq:cov1}
\begin{split}
cov(f(\th^{(j)}),g(\th^{(k)})) = & E[\tilde f(\th^{(j)})\tilde g(\th^{(k)})]=E[E[\tilde f(\th^{(j)})\tilde g(\th^{(k)})|\th^{(j)}]] \\
=& E[\tilde f(\th^{(j)})E[\tilde g(\th^{(k)})|\th^{(j)}]] = E_{\th^{(j)}}[\tilde f(\th^{(j)})\delta_{\th^{(j)}}\hat P_{j+1} \cdots \hat P_{k}\tilde g],
\end{split}
\end{equation}
where $\delta_{\th}$ is point mass at $\th$ and using our notation $\delta_{\th^{(j)}}\hat P_{j+1} \cdots \hat P_{k}$ corresponds to conditional distribution of $\th^{(k)}$ given fixed value of $\th^{(j)}$. \\
Using the general observation that for any two measures $\nu_1$ and $\nu_2$ and any bounded function $f$ the following inequality holds
\begin{equation}
|\nu_1 f - \nu_2 f| \leq 2 |f| \|\nu_1-\nu_2\|_{TV},
\end{equation}
we find that:
\begin{equation}
\begin{split}
|\delta_{\th^{(j)}}\hat P_{j+1} \cdots \hat P_{k}\tilde g| = & |\delta_{\th^{(j)}}\hat P_{j+1} \cdots \hat P_{k}\tilde g-0|=|\delta_{\th^{(j)}}\hat P_{j+1} \cdots \hat P_{k}\tilde g -\mu_k \tilde g|\\
=& |\delta_{\th^{(j)}}\hat P_{j+1} \cdots \hat P_{k}\tilde g - \mu_j \hat P_{j+1} \cdots \hat P_{k} \tilde g| \leq 2 |\tilde g| \|\delta_{\th^{(j)}}\hat P_{j+1} \cdots \hat P_{k} -
\mu_j \hat P_{j+1} \cdots \hat P_{k} \|_{TV} \\
\leq & 2|\tilde g| (1-\al^*)^{|k-j|}
\end{split}
\end{equation}
note that this result is for any $\th^{(j)}\in \Th$. 
Returning to \eqref{eq:cov1} we get that:
\begin{equation}
cov(\tilde f(\th^{(j)}),\tilde g(\th^{(k)})) \leq 2 |\tilde f||\tilde g| (1-\al^*)^{|k-j|}.
\end{equation}
Finally by triangular inequality $|\tilde f|\leq 2|f|$ for any $j=1,2,\cdots$ and similarly for $|\tilde g|$. The desired result follows immediately.
\end{proof}

\begin{proof}[Proof of Theorem~\ref{conv_result_2}]
Using our standard notation $\nu \hat P_0 \cdots \hat P_t f = E[f(\th^{(t)})]$, Theorem~\ref{conv_result_1}, Lemma~\ref{lemma_cov} and simple results for double sum of geometric series we get
\begin{equation}
\begin{split}
& E\left[ \left(\mu f - \frac{1}{M}\sum_{t=0}^{M-1} f(\th^{(t)})   \right)^2 \right] =  E\left[ \left(\mu f -\frac{1}{M}\sum_{t=0}^{M-1} \nu \hat P_0 \cdots \hat P_t f + \frac{1}{M}\sum_{t=0}^{M-1} \nu \hat P_0 \cdots \hat P_t f- \frac{1}{M}\sum_{t=0}^{M-1} f(\th^{(t)})   \right)^2 \right] \\
& = \left( \mu f -\frac{1}{M}\sum_{t=0}^{M-1} \nu \hat P_0 \cdots \hat P_t f \right)^2 + E\left[ \left(\frac{1}{M}\sum_{t=0}^{M-1} \nu \hat P_0 \cdots \hat P_t f- \frac{1}{M}\sum_{t=0}^{M-1} f(\th^{(t)})   \right)^2 \right] \\
&\leq \left( 2|f| \left( \frac{(1-(1-\al)^M)\|\mu-\nu\|_{TV}}{M\al} - \frac{\eps(1-(1-\al)^M)}{M\al^2} + \frac{\eps}{\al} \right) \right)^2 + \frac{1}{M^2} \sum_{j=0}^{M-1}\sum_{t=0}^{M-1} cov(f(\th^{(j)}),f(\th^{(t)})) \\
&\leq 4|f|^2 \left( \frac{(1-(1-\al)^M)}{M\al} - \frac{\eps(1-(1-\al)^M)}{M\al^2} + \frac{\eps}{\al} \right)^2 + \frac{8|f|^2}{M^2} \sum_{j=0}^{M-1}\sum_{t=0}^{M-1}(1-\al^*)^{|t-j|}\\
& = 4|f|^2 \left( \frac{(1-(1-\al)^M)}{M\al} - \frac{\eps(1-(1-\al)^M)}{M\al^2} + \frac{\eps}{\al} \right)^2 \\ 
& +  8|f|^2 \left( \frac{1}{M} +  \frac{2}{(\al^*)^2} \left( \frac{(1-\al^*)^{M+1}-(1-\al^*)}{M^2} + \frac{(1-\al^*)-(1-\al^*)^2}{M} \right) \right)
\end{split}
\end{equation}
Obtaining the desired result.
\end{proof}

\begin{proof}[Proof of Theorem~\ref{abc_conv}]
First by {\bf(B1)} - {\bf(B4)}, Theorem~\ref{hast-unif} guarantees uniform ergodicity of the exact chain $P$ with $\beta = \min_{\th \in \Th} \frac{q(\th)}{p(\th)h(\th)/c}$ where $c$ is the normalizing constant of the posterior. Note that $\beta>0$ since $\Th$ is compact, ratio is continuous  and never zero. Therefore $P$ also satisfies Doeblin Condition. Next from {\bf(B1)}, {\bf(B4)} and {\bf(B5)}, Theorem~\ref{knn} implies that $\sup_{\th \in \Th}\|\hat h(\th;\zn)-h(\th)\|\to 0$ with probability 1. Hence by Theorem~\ref{perturb_mcmc} perturbed kernel $\hat P$ can be made arbitrary close to the exact kernel $P$ for sufficiently large $N$. Note that total variation distance between $\hat P_N$ and $P$ decreases to zero as $N$ increases. Finally assumptions of Theorems~\ref{conv_result_1} and \ref{conv_result_2} follow trivially.
\end{proof}
\begin{proof}[Proof of Corollary~\ref{sbl_conv}]
First by {\bf(B1)}, {\bf(B2)}, {\bf(B3)}, {\bf(B4)} and {\bf(B8)}, Theorem~\ref{hast-unif} guarantees uniform ergodicity of the exact chain $P$ with 
$\beta = \min_{\th \in \Th} \frac{q(\th)}{p(\th)h(\th)/c}$ where $c$ is the normalizing constant of the posterior. Note that $\beta>0$ since $\Th$ is compact, ratio is continuous and never zero. 
Therefore $P$ satisfies Doeblin Condition. Next from {\bf(B1)}, {\bf(B5)}, {\bf(B6)} and {\bf(B7)}, Theorem~\ref{knn} implies that $\sup_{\th \in \Th}\|\hat h(\th;\zn)-h(\th)\|\to 0$ with probability 1. Hence by Theorem~\ref{perturb_mcmc} perturbed kernel $\hat P$ can be made arbitrary close to the exact kernel $P$ for sufficiently large $N$. Note that total variation distance between $\hat P_N$ and $P$ decreases to zero as $N$ increases. Finally assumptions of Theorems~\ref{conv_result_1} and \ref{conv_result_2} follow trivially.
\end{proof}
\section{Conclusion}
In this paper we proposed to speed up generic ABC-MCMC and BSL algorithms by storing past simulations. This approach significantly accelerates the process and can be very useful for models where simulation of a pseudo data set is computationally expensive or when large number of MCMC iterations is required. We presented theoretical arguments and necessary assumptions for convergence properties of the perturbed chain. The performance of these strategies were examined via a series of simulations under different models. All simulations summaries show that proposed methods significantly improve mixing and efficiency of the chain and at the same time produce as accurate and precise parameter estimates as generic samplers.       

\section*{Acknowledgments}
We thank Jeffrey Rosenthal and Stanislav Volgushev for constructive comments. RVC is grateful to the organizers of the BIRS workshop ``Validating and Expanding Approximate Bayesian Computation Methods" for creating a stimulating environment that generated ideas for this work. Finally, we  acknowledge funding support from NSERC of Canada.


\end{document}